\documentclass{article}
\usepackage{graphicx,psfrag,epsfig,epsf}
\usepackage{amsmath}
\usepackage{amsfonts}
\usepackage{amssymb}
\usepackage{verbatim}
\usepackage{latexsym}
\usepackage{color}
\usepackage[boxed]{algorithm2e}
\usepackage{diagrams}
\usepackage{tikz}
\usetikzlibrary{automata}
\usepackage{subfigure}
\usepackage[normalem]{ulem,mathtools}

\usepackage{verbatim}

\usepackage{latexsym}
\usepackage{stackrel}

\newcommand{\Post}{\operatorname{Post}}

\newcommand{\ctrl}{\mathrm{ctrl}}

\newcommand{\pc}{\mathrm{pc}}
\newcommand{\cp}{\mathrm{cp}}
\newcommand{\req}{\mathrm{req}}
\newcommand{\alt}{\mathrm{alt}}
\newcommand{\net}{\mathrm{net}}

\newcommand{\pd}{\mathrm{pd}}

\newtheorem{proposition}{Proposition}

\newtheorem{theorem}{Theorem}

\newtheorem{lemma}{Lemma}
\newtheorem{definition}{Definition}
\newtheorem{proof}{Proof}
\newtheorem{problem}{Problem}

\begin{document}



\title{{Design of Symbolic Controllers\\ for}  
	Networked Control Systems}

\author{Alessandro Borri, Giordano Pola and Maria Domenica Di Benedetto
\thanks{The research leading to these results has been partially supported by the Center of Excellence DEWS and received funding from the European Union Seventh Framework Programme [FP7/2007-2013] under grant agreement n. 257462 HYCON2 NoE.}
\thanks{Alessandro Borri is with the Istituto di Analisi dei Sistemi ed Informatica "A. Ruberti", Consiglio Nazionale delle Ricerche (IASI-CNR), 00185 Rome, Italy, alessandro.borri@iasi.cnr.it.}
\thanks{Giordano Pola and Maria Domenica Di Benedetto are with the Department of Information Engineering, Computer Science and Mathematics, Center of Excellence for Research DEWS, University of L{'}Aquila, 67100, L{'}Aquila, Italy, \{giordano.pola,mariadomenica.dibenedetto@univaq.it.\}}
}

\date{}

\maketitle

\thispagestyle{empty}

\pagestyle{empty}

\begin{abstract}                          
	Networked Control Systems (NCS) are distributed systems where plants, sensors, actuators and controllers communicate over shared networks.
	Non-ideal behaviors of the communication network include variable sampling/transmission intervals and communication delays, packet losses, communication constraints and quantization errors. NCS have been the object of intensive study in the last few years. However, due to the inherent complexity of NCS, current literature focuses on a subset of these non-idealities and mostly considers stability and stabilizability problems. Recent technology advances need different and more complex control objectives to be considered. In this paper we present first a general model of NCS, including most relevant non-idealities of the communication network; then, we propose a
	symbolic model approach to the control design with objectives expressed in terms of non-deterministic transition systems. The presented results are based on recent advances in symbolic control design of continuous and hybrid systems. An example in the context of robot motion planning with remote control is included, showing the effectiveness of the proposed approach.
\end{abstract}


\thispagestyle{empty}

\pagestyle{empty}


\section{Introduction}

Networked Control Systems (NCS) are complex, heterogeneous, spatially distributed systems where physical processes interact with distributed computing units through non-ideal communication networks. In the past, NCS were limited in the number of computing units and in the complexity of the interconnection network so that it was possible to obtain reasonable performance by aggregating subsystems that were locally designed and optimized. However the growth of complexity of the physical systems to control, together with the continuous increase in functions that these systems must perform, requires today to adopt a unified design approach where different disciplines (e.g. control systems engineering, computer science, software engineering and communication engineering) should contribute to reach new levels of performance.
The heterogeneity of the subsystems that are to be connected in an NCS make the control of these systems a hard but challenging task.
NCS have been the focus of much recent research in the control community:  Murray et al. in \cite{MurrayNCS} presented control over networks as one of the important future directions for control. Following \cite{HeemelsSurvey}, the most important non-idealities in the analysis of NCS are: (i) variable sampling/transmission intervals; (ii) variable communication delays; (iii) packet dropouts caused by the unreliability of the network; (iv) communication constraints (scheduling protocols) managing the possibly simultaneous transmissions over the shared channel; (v) quantization errors in the digital transmission with finite bandwidth.
There are two approaches to deal with such non-idealities: the \textit{deterministic} approach, which assumes worst-case (deterministic) bounds on the aforementioned imperfections, and the \textit{stochastic} approach, which provides a stochastic description of the non-ideal communication network. We focus on the deterministic methods, which can be further distinguished according to the modeling assumptions and the controller synthesis: a) the discrete-time approach (see e.g. \cite{CloostermanAutomatica10}, \cite{GarciaAutomatica07}) considers discrete-time controllers and plants; b) the sampled-data approach (see e.g. \cite{GaoAutomatica08}, \cite{NaghshtabriziSAGE10}) assumes discrete-time controllers and continuous-time (sampled-data) plants; c) the continuous-time (emulation) approach  (see e.g. \cite{HeemelsTAC10}, \cite{NesicTAC04})  focuses on continuous-time controllers and continuous-time (sampled-data) plants. Results obtained in the deterministic approach during the past few years are mostly about stability and stabilizability problems, see e.g.
\cite{NCSHesphana,HeemelsSurvey,NCS-HSCC2010}, and depend on the method considered and the assumptions on the non-ideal communication infrastructure.
In addition, current approaches in the literature take into account only a subset of these non-idealities. As reviewed in \cite{HeemelsSurvey}, for example, \cite{NesicTAC09} studies imperfections of type (i), (iv), (v), \cite{CloostermanAutomatica10}, \cite{NaghshtabriziCDC05}, \cite{NaghshtabriziSAGE10} consider simultaneously  (i), (ii), (iii), \cite{NesicTAC04} focuses on  (i), (iii), (iv), while \cite{GaoAutomatica08} manages (ii), (iii) and (v). Three types of non-idealities, namely (i), (ii), (iv), are considered for example in \cite{ChailletCDC08}, \cite{DonkersTAC11}, \cite{HeemelsTAC10}. In \cite{HeemelsCDC09}, the five non-idealities are dealt with but small delay and other restrictive assumptions are considered.
Finally, novel results in the stability analysis of NCS can be found in \cite{AlurTAC11}, \cite{AntunesTAC12}, \cite{BauerAutomatica12}, \cite{VanAutomatica12}. However, existing results do not address control design of NCS with complex specifications, as for example safety properties, obstacle avoidance, language and logic specifications.
This paper follows the deterministic approach and provides a framework for NCS control design where the {aforementioned non-idealities from (i) to (v)} 
can be taken into account. The proposed approach is based on the use of discrete abstractions of continuous and hybrid systems \cite{DiscAbs,paulo}, and follows the work in \cite{PolaAutom2008,MajidTAC11,GirardNA2013} based on the construction of \textit{symbolic models} for nonlinear and switched control systems. {As such, it offers a sound paradigm to solve control problems where software and hardware interact with the physical world, and to address a wealth of novel specifications that are difficult to enforce by means of conventional control design methods.}
Symbolic models  are abstract descriptions of complex systems where a symbol corresponds to an ``aggregate" of continuous states and a symbolic control label to an ``aggregate" of continuous control inputs. Several classes of dynamical and control systems that admit equivalent symbolic models have been identified in the literature. Within the class of hybrid automata we recall timed automata \cite{alur}, rectangular hybrid automata \cite{puri}, and o-minimal hybrid systems \cite{lafferriere,brihaye}. Early results for classes of control systems were based on dynamical consistency properties \cite{caines}, natural invariants of the control system \cite{koutsoukos}, $l$-complete approximations \cite{moor}, and quantized inputs and states \cite{forstner,BMP02}. Further works include results on controllable discrete-time linear systems \cite{TabuadaLTL}, piecewise-affine and multi-affine systems {\cite{polaTAC14}}, \cite{habets,BH06}, set-oriented discretization approach for discrete-time nonlinear optimal control problems \cite{junge1}, abstractions based on convexity of reachable sets \cite{gunther}, incrementally stable and incrementally forward complete nonlinear control systems with and without disturbances \cite{PolaAutom2008,MajidTAC11,PolaSIAM2009,BorriIJC2012}, switched systems \cite{GirardTAC2010} and time-delay systems \cite{PolaSCL10,PolaIJRNC2014}.
The interested reader is referred to \cite{GirardEJC11,paulo} for an overview on recent advances in this domain.

This paper addresses the control design of a fairly general model of NCS with complex specifications, and provides an extended version of the preliminary results published in \cite{BorriHSCC12,BorriCDC2012}. 
In particular, while in \cite{BorriHSCC12,BorriCDC2012} controllers are assumed to be static, we consider here general dynamic controllers. 

The main contributions are:

-- \emph{A general model of NCS}. We propose a general model of NCS, where the plant is a continuous-time nonlinear control system, the computing units are modelled by Moore machines, and the non-idealities introduced by the communication network include quantization errors, time-varying delay in accessing the network, time-varying delay in delivering messages through the network, limited bandwidth and packet dropouts.


-- \emph{Symbolic models for NCS}. We propose symbolic models that approximate NCS with arbitrarily good accuracy, { by using a novel notion,  introduced in this paper, called \emph{strong alternating approximate simulation}}. More specifically, under the assumption of existence of an incremental forward complete Lyapunov function for the plant of the NCS, we derive symbolic models approximating the NCS in the sense of strong alternating approximate simulation. {Stability of the open-loop NCS is not required}. 
{
	In some recent work \cite{ZamaniArxiv2014}, symbolic models for NCS are proposed, which, differently from our approach, are constructed on the basis of a symbolic model of the plant.}


-- \emph{Symbolic control design of NCS}. Building upon the {obtained} symbolic models, we address the NCS control design {problem, }where specifications are expressed in terms of transition systems. Given a NCS and a specification, a symbolic controller is derived such that the controlled system meets the specification \emph{in the presence of the considered non-idealities in the communication network}.


The paper is organized as follows. In Section \ref{sec:Notation} notation is introduced. In Section \ref{sec:modelingNCS} a model is proposed for a general class of nonlinear NCS. In  Section \ref{sec:SymbolicModels} symbolic models approximating NCS are derived. In Section \ref{sec:control} symbolic control design is addressed. An example of application of the proposed results is included in Section \ref{sec:example}. Finally, Section \ref{sec:conclusion} offers some concluding remarks. 
The Appendix recalls some technical notions that are instrumental in the paper.

\section{Notation and preliminary definitions}\label{sec:Notation}
\textit{Notation.} The symbols $\mathbb{N}$, $\mathbb{N}_0$, $\mathbb{Z}$, $\mathbb{R}$, $\mathbb{R}^{-}$, $\mathbb{R}^{+}$ and $\mathbb{R}_{0}^{+}$ denote the set of natural, nonnegative integer, integer, real, negative real, positive real, and nonnegative real numbers, respectively. The cardinality of a set $A$ is denoted by $|A|$. Given a set $A$ we denote $A^{2}=A\times A$ and $A^{n+1}=A\times A^{n}$ for any $n\in \mathbb{N}$. Given a pair of sets $A$ and $B$ and a relation $\mathcal{R}\subseteq A\times B$, the symbol $\mathcal{R}^{-1}$ denotes the inverse relation of $\mathcal{R}$, i.e. $\mathcal{R}^{-1}=\{(b,a)\in B\times A:( a,b)\in \mathcal{R}\}$; for $A' \subseteq A$ we define $\mathcal{R}(A')=\{b\in B| \exists a\in A' \text{ s.t. }(a,b)\in \mathcal{R} \}$ and for $B'\subseteq B$, $\mathcal{R}^{-1}(B')=\{a\in A| \exists b\in B' \text{ s.t. }(a,b)\in \mathcal{R} \}$. Given sets $A$, $B$ and $C$ and relations $\mathcal{R}_{ab}\subseteq A\times B$ and $\mathcal{R}_{bc}\subseteq B\times C$ we recall that the composition relation $\mathcal{R}=\mathcal{R}_{ab} \circ \mathcal{R}_{bc}\subseteq A\times C$ is defined as $\mathcal{R}_{ab} \circ \mathcal{R}_{bc}:=\{(a,c)\in A\times C | \exists b\in B \text{ s.t. }(a,b)\in \mathcal{R}_{ab} \wedge (b,c)\in \mathcal{R}_{bc}\}$. Note that, for any $A'\subseteq A$, $\mathcal{R}(A')=\mathcal{R}_{bc}(\mathcal{R}_{ab}(A'))$ and for any $C'\subseteq C$, $\mathcal{R}^{-1}(C')=\mathcal{R}^{-1}_{ab}(\mathcal{R}^{-1}_{bc}(C'))$. 
Given an interval $[a,b]\subseteq \mathbb{R}_0^+$,  we denote by $[a;b]$ (resp. $[a;b[$) the set $[a,b]\cap \mathbb{N}_0$ (resp. $[a,b[\cap \mathbb{N}_0$),  if $a\leq b$, and the empty set $\varnothing$ otherwise.
We denote the ceiling of a real number $x$ by $\lceil x \rceil=\min\{{n\in\mathbb{Z} \vert n\geq x}\}$. Given a vector $x\in\mathbb{R}^{n}$ we denote by $\Vert x\Vert$ the infinity norm and by $\Vert x\Vert_{2}$ the Euclidean norm of $x$. {Given any function  $f: D\rightarrow Y$ and any set $A\subseteq D$, we denote by $f(A)$ the image of the set $A$ through $f$, namely $f(A)=\{y\in Y : y=f(x), x \in D \}$.}

\textit{Preliminary definitions.}
A continuous function \mbox{$\gamma:\mathbb{R}_{0}^{+}\rightarrow\mathbb{R}_{0}^{+}$} is said to belong to class $\mathcal{K}$ if it is strictly increasing and \mbox{$\gamma(0)=0$}; a function $\gamma$ is said to belong to class $\mathcal{K}_{\infty}$ if \mbox{$\gamma\in\mathcal{K}$} and $\gamma(r)\rightarrow\infty$ as $r\rightarrow\infty$. 
{ Given $\varepsilon \in \mathbb{R}^+$ and $x=(x_1,x_2,...,x_n)\in\mathbb{R}^n$,} {the symbol $\mathcal{B}_{\varepsilon}(x)$ denotes the closed ball of radius $\varepsilon$ (in infinity norm) centered at $x$, i.e. {$\mathcal{B}_{\varepsilon}(x)=[-\varepsilon+x_1,x_1+\varepsilon] \times [-\varepsilon+x_2,x_2+\varepsilon] \times ... \times [-\varepsilon+x_n,x_n+\varepsilon]$}, while the symbol $\mathcal{B}_{[\varepsilon[}(x)$ denotes the set {$[x_1,x_1+\varepsilon[ \times [x_2,x_2+\varepsilon[ \times ... \times [x_n,x_n+\varepsilon[$.}}
{ Following \cite{PolaTAC12}, given any $\mu \in\mathbb{R}^+$ and any $x\in \mathbb{R}^n$, the symbol $[x]_{\mu}$ denotes the unique vector in $\mu \mathbb{Z}^n$ such that {$x \in \mathcal{B}_{[\mu[}([x]_{\mu})$}. As a consequence, $\Vert x-[x]_{\mu} \Vert \leq \mu$.} 
{
	Given $\mu\in\mathbb{R}^{+}$ and $A\subseteq \mathbb{R}^{n}$, we set $[A]_{\mu}:={\left\{b\in\mu \mathbb{Z}^n : b=[a]_{\mu}, a \in A\right\}}$ {and $\mathcal{B}_{[\mu[}(A)=\bigcup_{a\in A} \mathcal{B}_{[\mu[}(a) $}; if $B=\bigcup_{i\in [1;N]}A^{i}$ we set $[B]_{\mu}=\bigcup_{i\in [1;N]} ([A]_{\mu})^{i}$. 
	Consider a set $A$ given as a finite union of hyperrectangles, i.e. $A=\bigcup_{j\in [1;J]} A_j$, for some $J\in\mathbb{N}$, where {$A_j=\bigtimes_{k\in [1;n]} [\underline{a}_{j,k},\overline{a}_{j,k}[\subseteq \mathbb{R}^n$ with $\underline{a}_{j,k}<\overline{a}_{j,k}$, $\underline{a}_{j,k},\overline{a}_{j,k}\in \hat{\mu}_{A} {\mathbb{Z}}$} for some $\hat{\mu}_{A}\in\mathbb{R}^+$. By construction, for any integer $n_A \in \mathbb{N}$, by setting $\mu=\hat{\mu}_{A} / n_A$, we get that for any $a\in A$, $\Vert a-[a]_{\mu} \Vert \leq \mu$ and $[a]_{\mu} \in A${, implying $[A]_{\mu} \subseteq A$. 
	}}
	
	
	\section{Networked Control Systems\\ and Control Problem}\label{sec:modelingNCS}
	
	The class of NCS that we consider is depicted in Fig. \ref{NCSpic} 
	and is  inspired by the models reviewed in \cite{HeemelsSurvey}. The sub-systems composing the NCS are described hereafter.\\
	
	\begin{figure}
		\begin{center}
			\includegraphics[scale=1]{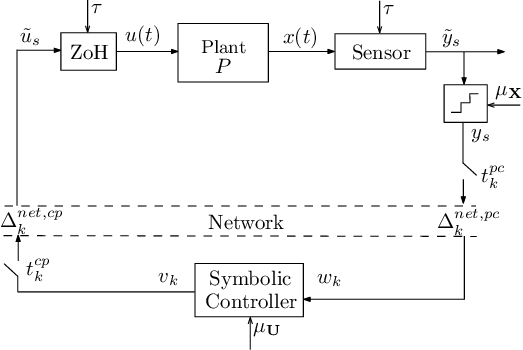}
			\caption{Networked Control System. A detailed description of the sub-systems depicted in this figure is reported in Section III.}
			\label{NCSpic}
		\end{center}
	\end{figure}
	
	\textit{\bf Plant.} The direct branch of the network includes the plant $P$ that is a nonlinear control system {of} the form:
	\begin{equation}
	\left\{
	\begin{array}
	{l}
	\dot{x}(t)=f(x(t),u(t)),\\
	x(t)\in 
	{ \mathbb{R}^{n} },\quad
	u(\cdot)\in \mathcal{U}, \quad {t\in \mathbb{R}_{0}^{+}},
	\end{array}
	\right.
	\label{NCSeq1}
	\end{equation}
	where $x(t)$ and $u(t)$ are the state and the control input at time $t$, 
	and $\mathcal{U}$ is the set of control inputs, defined as functions  
	from $\mathbb{R}^+_0$ to a finite non-empty set $\mathbf{U}\subset [\mathbb{R}^{m}]_{\mu_{\mathbf{U}}}$, for some $\mu_{{\mathbf{U}}}\in\mathbb{R}^{+}$, and constant in any interval $[s\tau,(s+1)\tau[$ with $s\in \mathbb{N}_0$ and for some given $\tau\in\mathbb{R}^+$, where $s$ is the index identifying the sampling interval (starting from $0$). 
	In the sequel we abuse notation by denoting the constant control input $u(t)={u}$ in the domain $[s\tau, (s+1)\tau[$ for all $s\in \mathbb{N}_0$ and for some $\tau\in\mathbb{R}^{+}$ by $u$. The function {$f:{ \mathbb{R}^{n} }\times \mathbf{U} \rightarrow \mathbb{R}^{n}$} is assumed to be Lipschitz on compact sets with respect to the first argument. In the sequel we denote by {$\mathbf{x}(t,x_0,u)$} the state reached by (\ref{NCSeq1}) at time $t$ under the control input $u$ from the state {$x_0$}.  
	We assume that the control system $P$ is forward complete in $\mathbb{R}^{n}$, namely that every trajectory {$\mathbf{x}(\cdot,{x(0)},u)$ of $P$} is defined on $[0,\infty[$. Sufficient and necessary conditions for a control system to be forward complete can be found in \cite{fc-theorem}. 
	\\  
	
	\textit{\bf Sensor.} On the right-hand side of the plant $P$ in Fig. \ref{NCSpic}, a sensor is placed. Since the sensor is physically connected to the plant, we assume that:\\
	
	(A.1) The sensor acts in time-driven fashion, it is synchronized with the plant and updates its output value at times that are integer multiples of $\tau\in\mathbb{R}^{+}$, i.e. 
	{$\tilde{y}_s=\mathbf{x}(s\tau,{x(0)},u)$.\\

		\textit{\bf Quantizer.} A quantizer follows the sensor. For simplicity, we assume that the quantizer is \textit{uniform}, with accuracy $\mu_{\mathbf{X}}\in \mathbb{R}^+
		$. 
		The role of the quantizer is: i) to discretize the continuous-valued sensor measurement sequence $\{\tilde{y}_s\}_{s \in{\mathbb{N}_0}}$ to get the quantized sequence $\{y_s\}_{s \in {\mathbb{N}_0}}$, with $y_s=[\tilde{y}_s]_{\mu_{\mathbf{X}}}$; ii) to encode the signals into digital messages 
		and to add overhead bits, resulting in the sequence of digital messages $\{\bar{y}_s\}_{s \in {\mathbb{N}_0}}$. 
		The transmission overhead takes into account the communication protocol, the packet headers, source and channel coding as well as data compression and encryption. We assume a fixed average relative overhead {$N_{\pc}^{+}\in]-1,+\infty[$} on each data bit ($N_{\pc}^{+}$ may be negative to include the case of  data compression). More precisely: 
		\\
		
		(A.2) $N_{\pc}^{+}$ bits are added per each bit of the digital signal encoding 
		${y}_s$, 
		{for all $s\in {\mathbb{N}_0}$}. \\
		
		
		\textit{\bf Network.} In the following, the index $k\in\mathbb{N}$ denotes the current iteration in the feedback loop. Due to the non-idealities of the network, not all the output samples can be transmitted through the network. We assume that only one output sample per iteration is sent. In particular, {$\{M_{k}\}_{k \in \mathbb{N}}\subseteq \mathbb{N}_0$} denotes the subsequence of the sampling intervals when the output samples are sent through the network, i.e.  at time $M_k \tau$ the digital message $\bar{y}_{M_{k}}$ encodes the output sample ${y}_{M_{k}}=[x(M_{k}\tau)]_{\mu_{\mathbf{X}}}$ and is sent (iteration $k$). We set $M_1=0$. 
		The communication network is characterized by the following features:\\
		
		\textit{Time-varying access to the network. } 
		The digital message $\bar{y}_{M_{k}}$ cannot be sent instantaneously to the network, because the communication channel is assumed to be a resource which is shared with other nodes or processes in the network. The policy by which a signal of a node is sent before or after a message of another  node is managed by the network scheduling protocol selected. We assume that:\\
		
		(A.3) The network waiting times $\Delta_{k}^{\req,\pc}$ in the plant-to-controller branch of the feedback loop are bounded, i.e. 
		$\Delta_{k}^{\req,\pc}\in [\Delta_{\min}^{\req}, \Delta_{\max}^{\req}]$, for some $\Delta_{\min}^{\req}$, $\Delta_{\max}^{\req}\in \mathbb{R}^+_0$. \\
		
		At time $t_{k}^{\pc}:=M_{k}\tau+\Delta_{k}^{\req,\pc}$, the message 
		$\bar{w}_k:= \bar{y}_{M_{k}}$ 
		is sent through the network. \\
		
		\textit{Limited bandwidth. } 
		In real applications, the capacity of the digital communication channel is 
		limited {and time-varying}. We denote by $B_{\min}$, $B_{\max}\in\mathbb{R}^{+}$,  
		{with $B_{\min} \leq B_{\max}$}, the minimum and maximum capacities of the channel (expressed in bits per second, bps). 
		In view of the binary coding and the transmission overhead (see Assumption (A.2)),  we assume that:\\
		
		(A.4) A delay $\Delta_{k}^{B,\pc}${$\in \mathbb{R}^+$,} due to the limited bandwidth, is introduced in the plant-to-controller branch of the feedback loop, {for all $k\in {\mathbb{N}}$}.\\ 

		\textit{Time-varying delivery of messages. } 
		The delivery of message {$\bar{w}_k$} may be subject to further delays, due to congestion phenomena in the network, etc. We assume that:\\
		
		(A.5) Network communication delays $\Delta_{k}^{\net,\pc}$ in the plant-to-controller branch of the feedback loop are bounded, i.e.
		$\Delta_{k}^{\net,\pc}\in [\Delta^{\net}_{\min},\Delta_{\max}^{\net}]$, for some $\Delta^{\net}_{\min}$, $\Delta_{\max}^{\net}\in\mathbb{R}_{0}^+$. \\
		
		\textit{Packet dropout. } 
		In real applications, one or more messages can be lost during the transmission, because of the unreliability of the communication channel. We assume that:\\
		
		(A.6) The maximum number of successive packet dropouts is $N_{\pd}$. \\
		
		\textit{\bf Symbolic Controller.} {After a finite number of possible retransmissions (see Assumption (A.6)),} 
		message {$\bar{w}_k$}  is decoded into the quantized sensor measurement {$w_k$} and reaches the controller. The symbolic controller $C$ is dynamic, non-deterministic, remote and asynchronous with respect to the plant and is expressed as a Moore machine:
		{
			\begin{equation}
			C: \left\{
			\begin{array}
			{l}
			\xi_{k}  \in f_C(\xi_{k-1},w_k), \quad {\xi_{k}  \in \Xi_C, k \in \mathbb{N}\setminus \{1\}},\\ 
			v_k  = h_C(\xi_k), \qquad {\hspace{5mm} v_k \in \mathbf{U}, \hspace{1.5mm} k \in \mathbb{N}},\\
			\xi_{1}\in \Xi^0_C, 
			\label{symbC}
			\end{array}
			\right.  \quad 
			\end{equation}
		}
		where $\Xi_C$ is the {finite} set of states of the controller, $\Xi^0_C{\color
			{violet}\subseteq \Xi_C}$ is the set of initial states of the controller, { $f_C$ {is a possibly partial function} $f_C: \Xi_C \times [{\mathbb{R}^n}]_{\mu_{{\mathbf{X}}}} \rightarrow 2^{\Xi_C}$ and $h_C : \Xi_C \rightarrow \mathbf{U}$. } 
		At each iteration $k$, the controller takes as input the measurement sample $w_k \in [{\mathbb{R}^n}]_{\mu_{{\mathbf{X}}}}${, updates its internal state to $\xi_k$ and returns the control sample $v_k = h_C(\xi_k)\in \mathbf{U}$} as output, which is synthesized by a computing unit that may be employed to execute several tasks. Note that, when $\Xi_C$ is a singleton set, $C$ becomes static. 
		The policy by which a computation is executed before or after another computation depends on the scheduling protocol adopted. We assume that: \\
		
		(A.7) The computation time $\Delta^{\ctrl}_k$ for the symbolic controller to return its output value {${v}_{k}$} is bounded, i.e.
		$\Delta_{k}^{\ctrl}\in [\Delta^{\ctrl}_{\min},\Delta_{\max}^{\ctrl}]$, for some $\Delta^{\ctrl}_{\min}$, $\Delta_{\max}^{\ctrl}\in\mathbb{R}_{0}^+$. \\
		
		The control sample {$v_k$} is encoded into a digital signal 
		and some overhead information is added to take into account the communication protocol, the packet headers, source and channel coding as well as data compression and encryption. The resulting message is denoted by {$\bar{v}_k$}. 
		We assume a fixed average relative overhead $N_{\cp}^{+}$ on each data bit, which may also be negative due to possible data compression. 
		{The following Assumptions (A.8) to (A.11), describing the non-idealities in the controller-to-plant branch of the network, correspond exactly to Assumptions (A.2) to (A.5), previously given for the plant-to-controller branch:\\
			
			(A.8) {$N_{\cp}^{+}\in]-1,+\infty[$} bits are added per each bit of {${v}_k$}.\\
			
			(A.9) Network waiting times $\Delta_{k}^{\req,\cp}$ in the controller-to-plant branch of the feedback loop are bounded, i.e. 
			$\Delta_{k}^{\req,\pc}\in [\Delta_{\min}^{\req}, \Delta_{\max}^{\req}]$. \\
			
			At time $t^{\cp}_k:=M_k \tau +\Delta_{k}^{\req,\pc}+\Delta_{k}^{B,\pc}+\Delta_{k}^{\net,\pc}+\Delta^{\ctrl}_{k}+\Delta_{k}^{\req,\cp}$, the message 
			$\bar{v}_k$ is sent. \\
			
			(A.10) A delay {{$\Delta_{k}^{B,\cp}\in \mathbb{R}^+$,}} due to the  limited bandwidth, is introduced in the controller-to-plant branch of the feedback loop.\\
			
			(A.11) Network communication delays ${\Delta_{k}^{\net,\cp}}$ in the controller-to-plant branch of the feedback loop are bounded, i.e.
			$\Delta_{k}^{\net,\cp}\in [\Delta^{\net}_{\min},\Delta_{\max}^{\net}]$. \\
			
			{We denote by 
				\[
				\Delta_k:=\Delta^{\req,\pc}_{k}\!+\Delta_{k}^{B,\pc}\!+\Delta_{k}^{\net,\pc}\!+\Delta_{k}^{\ctrl}\!+\Delta^{\req,\cp}_{k}\!+\Delta_{k}^{B,\cp}\!+\Delta_{k}^{\net,\cp}
				\]
				the 
				total delay induced by network and computing unit at iteration $k$, as a result of the assumptions above.}
			We can finally define 
			\begin{equation}
			N_k:=\left\lceil {\Delta_{k}} / {\tau} \right\rceil {\in \mathbb{N}}
			\label{eq:delta_to_N}
			\end{equation}
			as the \emph{discrete delay} induced by iteration $k$, expressed in terms of number of sampling intervals of duration $\tau$. 
			{{
					From the definitions of $M_k$ and $N_k$, we get $M_{k+1}=M_{k}+N_k$.} \\
				
				\textit{\bf ZoH.} {After a finite number of possible retransmissions (see Assumption (A.6)),} 
				message {$\bar{v}_k$}  is decoded into the control input $v_k$ and reaches the Zero-order-Holder (ZoH), placed on the left-hand side of the plant $P$ in Fig. \ref{NCSpic}.
			}
			We assume that:\\
			
			(A.12) The ZoH is updated at time $M_{k+1} \tau$ to the new value $v_k$, which is held exactly for one iteration, until a new control sample shows up, i.e. $u(t) = v_{k-1}$, $t\in[M_{k} \tau,M_{k+1} \tau[$. 
			At time $t=0$ a reference control input {$v_0:=\tilde{u}_{0}\in \mathbf{U}$} is held by the ZoH.\\

			{ 
				In the sequel we refer to the NCS model as $\Sigma${, which is also formally described in (\ref{sigma})}. 
				A trajectory of $\Sigma$ is a function $x:{\mathbb{R}_0^+}\rightarrow \mathbb{R}^n$ 
				satisfying (\ref{sigma}). 
			}} 
			\begin{figure*}[!t]
				\normalsize
				\newcounter{MYtempeqncnt}
				\setcounter{MYtempeqncnt}{\value{equation}}
				\setcounter{equation}{3}
				\begin{equation}
				\label{sigma}
				\Sigma:
				\left\{
				\begin{array}
				{ll}
				\text{Iteration delay:}
				&
				N_{k}=\left\lceil \frac{\Delta_{k}}{\tau} \right\rceil , {\Delta}_{k} \in {\mathbb{R}^+}, k\in\mathbb{N},\\
				\text{Sampling/holding time sequence:}
				&
				\left\{
				\begin{array}
				{l}
				M_{k+1}=M_{k}+N_k, k\in\mathbb{N},\\
				M_{1}=0,
				\end{array}
				\right.\\
				{\text{ZoH:}}
				&
				{\left\{
					\begin{array}
					{l}
					u(t)=\sum{_{k=1}^{\infty}} 
					{v}_{k-1}\mathbf{1}_{[M_k \tau,M_{k+1}\tau[}(t),{ t\in \mathbb{R}_0^+},\\
					v_0=\tilde{u}_0  \quad \text{ given},
					\end{array}
					\right.}
				\\
				\text{Plant:}
				&
				\left\{
				\begin{array}
				{l}
				\dot{x}(t)=f(x(t),u(t)),\\
				x(t)\in { \mathbb{R}^{n}}, \quad 
				u(\cdot)\in \mathcal{U}, \quad { t\in \mathbb{R}_0^+}, 
				\end{array}
				\right.
				\\
				{\text{Sensor:}}
				&
				\tilde{y}_s=
				\mathbf{x}(s \tau,{x(0)},u){\in \mathbb{R}^n}, 
				s\in {\mathbb{N}_0},\\
				\text{Quantizer:}
				&
				y_{s}=[\tilde{y}_s]_{\mu_{\mathbf{X}}}{ 
				}
				, s\in {\mathbb{N}_0},
				\\
				{\text{Switch:}}
				&
				{w}_{k}=y_{s}, s=M_{k}, 
				{k\in\mathbb{N},}
				\\
				\text{Controller:}
				&
				\left\{
				\begin{array}{l}
				{\xi_{k}  \in f_C(\xi_{k-1},w_k),  \quad {\xi_{k}  \in \Xi_C, k \in \mathbb{N}\setminus \{1\}}, 
				}  
				\\
				{v_k  = h_C(\xi_k), \qquad {\hspace{5mm} v_k \in \mathbf{U}, \hspace{1.5mm} k \in \mathbb{N}}, 
				}
				\\
				\xi_{1}\in \Xi^0_C.
				\end{array}
				\right.
				\end{array}
				\right.
				\end{equation}
				\setcounter{equation}{4}
				\hrulefill
				\vspace*{4pt}
			\end{figure*}
			Due to possible different realizations of the non-idealities and the non-deterministic controller, the NCS $\Sigma$ is non-deterministic. 
			Note that the definition of NCS given in this section allows taking into account different scheduling protocols and communication constraints: any  protocol or set of protocols  satisfying Assumptions (A.2--A.5), (A.6) and (A.8--A.11), {such as Controller Area Network (CAN) \cite{CAN} and Time Triggered Protocol (TTP)  \cite{TTP} used in vehicular and industrial applications,}  can be used.\
			
			We {conclude this section by introducing} the control problem that we address in this paper. 
			We consider a control design problem where the NCS $\Sigma$ has to satisfy a 
			specification $Q$, given in terms of  a non-deterministic transition system, up to a desired {accuracy} $\varepsilon$, while being robust with respect to the non-idealities of the communication network. More formally:
			
			\begin{problem}
				\label{problem0}
				Consider a specification $Q$ expressed in terms of a {finite} collection of transitions $T_{Q}\subseteq X_{Q}\times X_{Q}$, with {$X_{Q}\subseteq\mathbb{R}^n$}, and let $X_{Q}^{0}\subseteq X_{Q}$ be a set of initial states. For any desired {accuracy} $\varepsilon\in\mathbb{R}^+$, find 
				a quantization parameter $\mu_{\mathbf{X}}\in\mathbb{R}^+$,  
				a set of initial states  $\mathbf{X}_0$ of the plant and a {symbolic} controller $C$ in the form of (\ref{symbC}) such that, 
				{ for any sequence $\{\tilde{y}_s\}_{s\in\mathbb{N}_0}$ generated by the NCS $\Sigma$ in (\ref{sigma}) with $\tilde{y}_0\in {\mathbf{X}}_0$,
				}
				there exists a sequence $\{x_Q^s\}_{s \in \mathbb{N}_0}$ with $x_Q^0 \in X_{Q}^{0}$ such that, for any discrete-time $s\in \mathbb{N}_0$, the following conditions hold:
				\begin{itemize}
					\item [1)] $(x_Q^s,x_Q^{s+1})\in T_Q$; 
					\item [2)] $\Vert \tilde{y}_{s} - x_Q^s \Vert \leq \varepsilon$.
				\end{itemize}
			\end{problem}
			\section{Symbolic Models for NCS}\label{sec:SymbolicModels}
			
			In this section we propose symbolic models that approximate NCS with arbitrarily good accuracy{, which} is instrumental to give {in Section \ref{sec:control}} the solution to Problem \ref{problem0}. 
			
			{
				We start by providing tighter bounds on the delay defined in Section \ref{sec:modelingNCS}, depending on the particular specification considered. Consider a set $\mathbf{X}$, 
				with  $\mathcal{B}_{\varepsilon}(X_Q)\subseteq\mathbf{X}\subseteq \mathbb{R}^n$, given as a finite union of hyperrectangles $\mathbf{X}=\bigcup_{j\in [1;J]} \mathbf{X}_j$, for some $J\in\mathbb{N}$, each in the form ${\mathbf{X}_j=}\bigtimes_{k\in [1;n]} [\underline{x}_{j,k},\overline{x}_{j,k}[$, with $\underline{x}_{j,k}<\overline{x}_{j,k}$, $\underline{x}_{j,k},\overline{x}_{j,k}\in \hat{\mu}_{\mathbf{X}} {\mathbb{Z}}$ for some $\hat{\mu}_{\mathbf{X}}\in\mathbb{R}^+$. 
				The property $\mathcal{B}_{\varepsilon}(X_Q)\subseteq\mathbf{X}$ and condition 2) in Problem \ref{problem0} imply that, if a controller $C$ in the form \eqref{symbC} solves Problem \ref{problem0}, then the corresponding sensor measurements $\tilde{y}_s $ belong to the bounded set $\mathbf{X}$ for all $s\in \mathbb{N}_0$. As a consequence, it is possible to provide an upper-bound on the length of the digital messages  encoding sensor measurements and, in turn, some uniform bounds on the delay $\Delta_k$ induced by each network iteration.} {In particular:
				\begin{itemize}
					\item  Assumption A.2) implies that the number of bits of  
					message $\bar{y}_s$ is {bounded by} $\lceil(1+N_{\pc}^{+}) \lceil\log_{2} \vert [\mathbf{X}]_{\mu_{{\mathbf{X}}}} \vert \rceil\rceil$, {for all $s\in {\mathbb{N}_0}$};
					\item  from Assumption A.4), one has $\Delta^{B,\pc}_{k}\in [\Delta_{\min}^{B,\pc},\Delta_{\max}^{B,\pc}]$, with $\Delta_{\min}^{B,\pc}=\lceil(1+N_{\pc}^{+}) \lceil\log_{2} \vert [\mathbf{X}]_{\mu_{{\mathbf{X}}}} \vert \rceil \rceil/ B_{\max}$ and $\Delta_{\max}^{B,\pc}=\lceil(1+N_{\pc}^{+}) \lceil\log_{2} \vert [\mathbf{X}]_{\mu_{{\mathbf{X}}}} \vert \rceil \rceil/ B_{\min}$;
					\item  Assumption (A.8) implies that  the number of bits of $\bar{v}_k$ is {bounded by} $\lceil(1+N_{\cp}^{+}) \lceil\log_{2} \vert \mathbf{U} \vert \rceil \rceil$;
					\item from Assumption (A.10), one has \\$\Delta_{k}^{B,\cp}\in [\Delta_{\min}^{B,\cp} ,\Delta_{\max}^{B,\cp}]$, with $\Delta_{\min}^{B,\cp}=\lceil(1+N_{\cp}^{+}) \lceil\log_{2} \vert \mathbf{U} \vert \rceil \rceil / B_{\max}$ and $\Delta_{\max}^{B,\cp}= \lceil (1+N_{\cp}^{+}) \lceil\log_{2} \vert \mathbf{U} \vert \rceil \rceil / B_{\min}$.
				\end{itemize}}

				
				
				In the absence of packet dropouts,  one has ${\Delta}_{k}\in [\bar{\Delta}_{\min},\bar{\Delta}_{\max}]$, where $\bar{\Delta}_{\min}$, $\bar{\Delta}_{\max} \in \mathbb{R}^+$ are the minimum and maximum delays computed according to the {given} assumptions (excluding (A.6)),  as 
				\[\bar{\Delta}_{\min} := \Delta_{{\min}}^{B,\pc}+\Delta_{\min}^{\ctrl}+\Delta_{{\min}}^{B,\cp}+2\Delta^{\req}_{\min}+2\Delta_{\min}^{\net},\]
				\[\bar{\Delta}_{\max} := \Delta_{{\max}}^{B,\pc}+\Delta_{\max}^{\ctrl}+\Delta_{{\max}}^{B,\cp}+2\Delta^{\req}_{\max}+2\Delta_{\max}^{\net}.\]
				
				{In presence of packet dropouts,} under Assumption (A.6) and following the so-called \emph{emulation approach}, reformulating {them} in terms of additional delays, see e.g. \cite{HeemelsSurvey}, it is readily seen that iteration $k$ introduces a time-varying delay ${\Delta}_{k} \in [{\Delta}_{\min},{\Delta}_{\max}]$ {in \eqref{sigma}}, with 
				$\Delta_{\min} = \bar{\Delta}_{\min}$ and $\Delta_{\max} = (1+N_{\pd})\bar{\Delta}_{\max}$,
				where $N_{\pd}$ is the maximum number of subsequent packet dropouts.	
				{Consequently, discrete delays $N_k$ in \eqref{eq:delta_to_N} will be bounded as follows:
					\begin{equation}
					N_k\in [N_{\min};N_{\max}] \quad  \forall k\in \mathbb{N},
					\label{eq:N_k}
					\end{equation}
					with bounds given by: 
					\begin{equation}
					N_{\min}=\left\lceil {\Delta_{\min}} / {\tau} \right\rceil\in\mathbb{N}, \qquad N_{\max}=\left\lceil {\Delta_{\max}} / {\tau} \right\rceil\in\mathbb{N}.
					\label{eq:N_min_max}
					\end{equation}
					
				}

				
				
				{We are now ready to use the notion of system as a unified mathematical framework to describe NCS.
				}
				
				
				
				\begin{definition}
					\label{ssystem}
					\cite{paulo}
					A system is a sextuple \[S=(X,X_{0},U,\rTo,Y,H)\] consisting of 
					a set of states $X$, a set of initial states $X_{0}\subseteq X$, a set of inputs $U$, a transition relation $\rTo \subseteq X\times U\times X$, a set of outputs $Y$ and an output function $H:X\rightarrow Y$.
					A transition $(x,u,x^{\prime})\in\rTo$ of $S$ is denoted by $x\rTo^{u}x^{\prime}$. For such a transition, state $x^{\prime}$ is called a $u$-successor or simply a successor of state $x$. We denote by $\Post_u(x)$ the set of $u$-successors of a state $x$ and by $U(x)$ the set of inputs $u\in U$ for which $\Post_u(x)$ is nonempty.
				\end{definition}
				
				System $S$ is said to be \textit{symbolic} (or  \textit{finite}), if $X$ and $U$ are finite sets, \textit{(pseudo)metric}, if the output set $Y$ is equipped with a (pseudo)metric $d:Y\times Y\rightarrow\mathbb{R}_{0}^{+}$, \textit{deterministic}, if for any $x\in X$ and $u\in U$ there exists at most one state $x^{\prime}\in X$ such that $x \rTo^{u} x^{\prime}$, \textit{non-blocking}, if $U(x)\neq \varnothing$ for any $x\in X$.
				The evolution of systems is captured by the notions of state and output runs. { A state run of $S$ is a {possibly infinite} sequence $\{x_i\}$  such that { $x_0\in X_0$ and}, for any $i$, there exists $u_{i}\in U$ for which $x_{i} \rTo^{u_{i}} x_{i+1}$. An output run is a {possibly infinite} sequence $\{y_i\}$ such that there exists a state run $\{x_i\}$ with $y_i=H(x_i)$ for any $i$.} 
				In order to give a representation of NCS in terms of systems, we first need to provide an equivalent formulation of NCS.  
				Given the NCS $\Sigma$, 
				consider the NCS $\Sigma_d$ depicted in Fig. 2 and with evolution formally specified by equations (\ref{sigma2}). 
				\begin{figure*}[!t]
					\normalsize
					\setcounter{MYtempeqncnt}{\value{equation}}
					\setcounter{equation}{6}
					\begin{equation}
					\label{sigma2}
					\hspace{-7mm}\Sigma_d:
					\left\{
					\begin{array}
					{ll}
					\hspace{-2mm}\bar{\Sigma}_d:
					&
					\hspace{-15mm}\left\{
					\begin{array}
					{ll}
					\text{Iteration delay:}
					&
					N_{k} \in {\mathbb{N}}, k\in\mathbb{N},\\
					\text{Sampling/holding time sequence:}
					&
					\hspace{-3mm}\left\{
					\begin{array}
					{l}
					M_{k+1}=M_{k}+N_k, k\in\mathbb{N},\\
					M_{1}=0,
					\end{array}
					\right.\\
					{\text{Sampled-data control system $P_d$: }}
					&  
					\hspace{-3mm}\begin{cases}
					z_{s+1}=\bar{f}(z_s,{v_{k-1}}){=\mathbf{x}(\tau,z_s,{v_{k-1}})}{\in \mathbb{R}^n}, s\in { 
						[M_{k};M_{k+1}[,  
					}\text{ }{ k\in\mathbb{N},}\\
					\tilde{y}_s=z_s, 
					s\in { 
						\mathbb{N}_0,}\\
					{z_0{=x(0)}, \qquad v_0=\tilde{u}_0} \quad \text{  given,}
					\end{cases}   
					\end{array}
					\right.
					\\
					\text{Quantizer:}
					&
					y_{s}=[\tilde{y}_s]_{\mu_{\mathbf{X}}}{ 
					}, s\in {\mathbb{N}_0},\\
					{\text{Switch:}}
					&
					{w}_{k}=y_{s}, s=M_{k}, 
					{k\in\mathbb{N},}
					\\
					\text{Controller:}
					&
					\left\{
					\begin{array}
					{l}
					{\xi_{k}  \in f_C(\xi_{k-1},w_k),  \quad {\xi_{k}  \in \Xi_C, k \in \mathbb{N}\setminus \{1\}},
					}  
					\\
					{v_k  = h_C(\xi_k), \qquad {\hspace{5mm} v_k \in \mathbf{U}, \hspace{1.5mm} k \in \mathbb{N}}, 
					}\\
					\xi_{1}\in \Xi^0_C.
					\end{array}
					\right.
					\end{array}
					\right.
					\end{equation}
					\setcounter{equation}{7}
					\hrulefill
					\vspace*{4pt}
				\end{figure*}
				In equations (\ref{sigma2}), we replace the interconnected blocks ZoH, Plant and Sensor of (\ref{sigma}) by the nonlinear \emph{sampled-data control system} $P_d$, where
				\[
				{
					\bar{f}({x},{u}):=\mathbf{x}(\tau,{x},u), 
				}
				\]
				{for any ${x}\in{\mathbb{R}^n}$ and $u\in \mathbf{U}$,}
				{which is} the time discretization of the plant $P$ with sampling time $\tau$. 
				A sequence $\{z_s \}_{s \in {\mathbb{N}_0}}$ satisfying (\ref{sigma2}) for 
				some sequence $\{v_k \}_{k \in {\mathbb{N}_0}}$ is called a \textit{trajectory} of ${ \Sigma_d}$.  
				We stress that control sample $v_{k-1}$, designed at iteration $k-1$, is applied to the plant $P_d$ at iteration $k$; this delay in the iteration index translates into a physical delay $N_{k-1} \tau$ {for the} application of the new control sample; indeed, sample $v_{k-1}$ is applied at time $t=M_{k} \tau$, with $M_{k}=M_{k-1} + N_{k-1}$. {We give the following  result that is instrumental for the further developments.}

				
				\begin{proposition}
					\label{TeoEquivalence} \textit{ }
					\begin{itemize}
						\item[a)] 
						{ For any trajectory $x:{\mathbb{R}^+_0}\rightarrow \mathbb{R}^n$ of $\Sigma$ 
							there exists a trajectory $\{z_s \}_{s \in \mathbb{N}_0}$ of $\Sigma_d$ such that $z_s=x(\tau s)$ for all $s \in {\mathbb{N}_0}$;}
						\item[b)] 
						{ For any trajectory 
							$\{z_s \}_{s \in {\mathbb{N}_0}}$ of $\Sigma_d$ there exists a trajectory 
							$x:{\mathbb{R}^+_0}\rightarrow \mathbb{R}^n$ of $\Sigma$ such that $z_s=x(\tau s)$ for all $s \in {\mathbb{N}_0}$.}
					\end{itemize}
				\end{proposition}

				We now have all the ingredients to provide a system representation of $\bar{\Sigma}_d$. 
				
				\begin{figure}
					\begin{center}
						\includegraphics[scale=0.8]{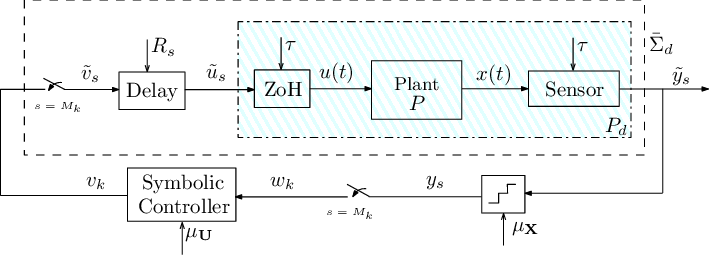}
						\caption{Illustration of $\Sigma_d$, which is formally described by the equations in (\ref{sigma2}). The sequence $\{\tilde{y}_s\}_{ s \in {\mathbb{N}_0}}$ includes all output samples of the sampled-data control system $P_d$. At each iteration $k$, the quantized output  ${w}_{k}=y_{s}=[\tilde{y}_s]_{\mu_{\mathbf{X}}}$  for  $s=M_{k}$ reaches the controller and a control input value $v_k$ is computed. The Delay block takes into account  the total delay ${N}_{k}$ of the NCS loop at iteration $k$ and the fact that the control value computed at iteration $k$ becomes active at iteration $k+1$.
						}
						\label{NCSpic3}
					\end{center}
				\end{figure}
				
				\begin{definition}
					Given $\bar{\Sigma}_d$ {in (\ref{sigma2}), with $N_k$ satisfying \eqref{eq:N_k}}, define the system 
					\[
					S(\bar{\Sigma}_d)=(X_{\tau},X_{0,\tau},\mathbf{U},\rTo_{\tau},Y_{\tau},H_{\tau}), 
					\]
					where
					\begin{itemize} 
						\item $X_{\tau}= \left( {\mathbb{R}^n} \times \{\tilde{u}_0\} \right) \cup \{ \left(x_{1},...,x_{N},\bar{u}\right)  \in {\mathbb{R}^{nN}} \times \mathbf{U}: \exists \underline{u}\in \mathbf{U} \text{ s.t. }  
						x_{i+1} =\bar{f}(x_{i},\underline{u}) \text{ }   \forall i$ $\in [1;N-1] 
						$, $N\in  [N_{\min};N_{\max}]$\};
						\item $X_{0,\tau}={\mathbb{R}^n} \times \{\tilde{u}_0\}$; 
						\item $x^{1}=({x}_0,\tilde{u}_0)\rTo^{u}_{\tau} x^{2}\!=\!\left(x_{1}^{2},...,x_{N_2}^{2},\bar{u}^2\right)$, if $x^1\in X_{0,\tau}$, $\bar{u}^{2}  =u$, $x_{1}^{2}  =\bar{f}(x_0,\tilde{u}_{0})$ and $x_{i+1}^{2}  =\bar{f}(x_{i}^{2},\tilde{u}_{0})$ for  $i\in [1;N_2 -1]$, $N_2\in [N_{\min};N_{\max}]$;
						\item $x^{1}\!=\!\left(x_{1}^{1},...,x_{N_1}^{1},\bar{u}^1\right)\rTo^{u}_{\tau} x^{2}\!=\!\left(x_{1}^{2},...,x_{N_2}^{2},\bar{u}^2\right)$, if $\bar{u}^{2}  =u$, $x_{1}^{2}  =\bar{f}(x^1_{N_{1}},\bar{u}^{1})$ and $x_{i+1}^{2}  =\bar{f}(x_{i}^{2},\bar{u}^{1})$ for  $i\in [1;N_2 -1]$, $N_1,N_2\in  [N_{\min};N_{\max}]$;
						\item $Y_{\tau}={\mathbb{R}^n \cup (\bigcup_{N\in [N_{\min};N_{\max}]} \mathbb{R}^{nN})}$;
						\item {$H_{\tau}\left({x}_0,\tilde{u}_0\right)=x_0$ for all ${x}_0\in {\mathbb{R}^n}$};
						\item $H_{\tau}\left(x_{1},x_{2},...,x_{N},\bar{u}\right)=\left(x_{1},x_{2},...,x_{N}\right)$, for all $\left(x_{1},x_{2},...,x_{N},\bar{u}\right)\in X_{\tau}$, $N\in [N_{\min};N_{\max}]$.
					\end{itemize}
					\label{def:system_Sigma}
				\end{definition}
				
				Note that $S(\bar{\Sigma}_d)$ is non-deterministic because, depending on the values of $N_{2}$ in the transition relation, {multiple} $u$-successor{s} of $x^{1}$ exist. 
				{ System $S(\bar{\Sigma}_d)$ can be regarded as a pseudometric system with the pseudometric $d_{Y_{\tau}}$ on $Y_{\tau}$ naturally induced by the metric { $d(x_{1},x_{2})=\Vert x_{1} - x_{2}\Vert$ on $\mathbb{R}^n$},  as follows. Given any $x^{i}=(x_{1}^{i},x_{2}^{i},...,x_{N_{i}}^{i},\bar{u}^{i})$, $i=1,2$, we set 
					\[
					d_{Y_{\tau}}(H_{\tau}(x^{1}),H_{\tau}(x^{2}))\!\!=\!\!
					\begin{cases}
					\max_{i\in[1;N_1]}\Vert x^{1}_{i}-x^{2}_{i}\Vert, & \hspace{-2mm} \text{if $N_{1}= N_{2}$};\\
					+\infty, & \hspace{-2mm}  \text{otherwise}.
					\end{cases}
					\]}
				Since the state vectors of $S(\bar{\Sigma}_d)$ are built from the trajectories of $P_d$ in $\bar{\Sigma}_d$, it is readily seen that:
				\begin{proposition}
					For any trajectory $\{z_s \}_{{s \in {\mathbb{N}_0}}}$ of 
					{ $\Sigma_d$,} {with $N_k$ satisfying \eqref{eq:N_k},} 
					there exists a state run
					\begin{equation}
					\underbrace{({x(0)},\tilde{u}_0)}_{x^0} \rTo^{\tilde{u}_{1}}  \underbrace{(\bar{x}^{1},\tilde{u}_{1})}_{x^1}  \rTo^{\tilde{u}_{2}}  \underbrace{(\bar{x}^{2},\tilde{u}_{2})}_{x^2}  \rTo^{\tilde{u}_{3}} \,\, {...}\,
					\label{cond1}
					\end{equation}
					of $S(\bar{\Sigma}_d)$ such that: 
					\begin{equation}
					\begin{array}
					{rclclcl}
					\{ {{x(0)}} & , & \underbrace{\bar{x}^{1}_{1}, {...}, \bar{x}^{1}_{N_{1}}}_{\bar{x}^{1}} & , & \underbrace{\bar{x}^{2}_{1}, {...}, \bar{x}^{2}_{N_{2}}}_{\bar{x}^{2}} & , & {...} \} = \{z_s \}_{{s \in {\mathbb{N}_0}}}.
					\end{array}
					\label{cond2}
					\end{equation}
					Conversely, for any state run (\ref{cond1}) of $S(\bar{\Sigma}_d)$, there exists a trajectory $\{z_s \}_{{s \in {\mathbb{N}_0}}}$ of  
					{ $\Sigma_d$} 
					such that (\ref{cond2}) holds.
					\label{theo:theo1}
				\end{proposition}

				\medskip
				
				Although system $S(\bar{\Sigma}_d)$ contains all the information of the NCS available at the sensor, it is not a finite model. Hereafter, 
				we illustrate the construction of 
				symbolic models that approximate possibly unstable NCS in the sense of strong alternating approximate simulation, whose definition is formally introduced in the Appendix. Our results rely on the assumption of existence of an incremental forward complete ($\delta$-FC) Lyapunov function for the plant of the NCS. 
				More formally:
				
				\begin{definition}
					\label{TH-IFC}
					\cite{MajidTAC11}
					A { continuously differentiable} function $V:{\mathbb{R}^n \times \mathbb{R}^n } \rightarrow \mathbb{R}_0^+$ is a $\delta$-FC Lyapunov function for the plant control system of the NCS if there exist a real number $\lambda\in\mathbb{R}$ and $\mathcal{K}_{\infty}$ functions $\underline{\alpha}$ and $\overline{\alpha}$ such that, for any $x_{1},x_{2}\in {\mathbb{R}^n}$ and any $u\in \mathbf{U}$, the following conditions hold:
					\begin{itemize}
						\item
						[(i)] $\underline{\alpha}(\Vert{x_{1}-x_{2}}\Vert)\leq V(x_{1},x_{2})\leq\overline{\alpha}(\Vert{x_{1}-x_{2}}\Vert)$,
						\item
						[(ii)] 
						$\frac{\partial{V}}{\partial{x_{1}}}(x_1,x_2)\, f(x_{1},u)\!+\! \frac{\partial{V}}{\partial{x_{2}}}(x_1,x_2)\, f(x_{2},u) \! \leq \! \lambda V(x_{1},x_{2})$.
					\end{itemize}
				\end{definition}
				
				We refer the interested reader to \cite{MajidTAC11} for further details on this notion. 
				In the following, we suppose the existence of a $\delta$-FC Lyapunov function $V$ for the control system $P$ in the NCS $\Sigma$ {and}
				of a $\mathcal{K}_{\infty}$ function $\gamma$ 
				such that $V(x,x^{\prime})-V(x,x^{\prime \prime})\leq\gamma(\Vert{x^{\prime}-x^{\prime \prime}}\Vert)$, for every $x,x^{\prime},x^{\prime\prime}\in {\mathbb{R}^n}$. We assume {without loss of generality} that $V$ is symmetric, i.e. $V(x_1,x_2)=V(x_2,x_1)$ for all $x_1,x_2 \in {\mathbb{R}^n}$. 
				
				\begin{definition}
					{Given $\bar{\Sigma}_d$ in (\ref{sigma2}), with $N_k$ satisfying \eqref{eq:N_k},}	define the system 
					\[
					S_{\ast}(\bar{\Sigma}_d):=(X_{\ast},X_{0,\ast},\mathbf{U},\rTo_{\ast},Y_{\ast},H_{\ast}),
					\]
					where  
					\begin{itemize}
						\item $X_{\ast}= \left( [{\mathbb{R}^n}]_{\mu_{{\mathbf{X}}}} \times \{\tilde{u}_0\} 
						\right) \cup \{ \left(x^{\ast}_{1},x^{\ast}_{2},...,x^{\ast}_{N}, \bar{u}_{\ast}\right)  \in [{\mathbb{R}^{nN}}]_{\mu_{{\mathbf{X}}}} \times \mathbf{U} : \exists \underline{u}_{\ast}\in \mathbf{U}  \text{ s.t. }$  
						${
							V([\bar{f}(x^{\ast}_{i},\underline{u}_{\ast})]_{\mu_{\mathbf{X}}},x^{\ast}_{i+1}) \leq }$ ${ 
							(e^{\lambda \tau}+2) \gamma(\mu_{\mathbf{X}}),}$ ${\forall i\in [1;N-1] }$, $N\in [N_{\min};N_{\max}]$\};
						\item $X_{0,\ast}=[{\mathbb{R}^n}]_{\mu_{{\mathbf{X}}}} \times \{\tilde{u}_0\}$, 
						\item $x^{1}=\left(x_0,\tilde{u}_0\right)\rTo^{u_{\ast}}_{\ast} x^{2}=\left(x_{1}^{2},...,x_{N_2}^{2},\bar{u}_{\ast}^{2}\right)$, if 
						{
							$x^{1}\in X_{0,\ast}$, $\bar{u}_{\ast}^{2}  =u_{\ast}$, }
						$N_2\in[N_{\min};N_{\max}]$,
						and
						\begin{align}
						\!\begin{cases}\!
						V([\bar{f}(x_0,\tilde{u}_0)]_{\mu_{\mathbf{X}}},x_{1}^{2}) \! \leq 
						(e^{\lambda \tau}\!+\!2) \gamma(\mu_{\mathbf{X}}),\\
						V([\bar{f}(x_{i}^{2},\tilde{u}_0)]_{\mu_{\mathbf{X}}},x_{i+1}^{2}) \! \leq 
						(e^{\lambda \tau}\!+\!2) \gamma(\mu_{\mathbf{X}})
						,   i \! \in\! [1; N_2-1];
						\end{cases}
						\label{eq:trans_rel_symb_0}
						\end{align}
						\item $x^{1}=\left(x_{1}^{1},...,x_{N_1}^{1},\bar{u}_{\ast}^{1}\right)\rTo^{u_{\ast}}_{\ast} x^{2}=\left(x_{1}^{2},...,x_{N_2}^{2},\bar{u}_{\ast}^{2}\right)$, if 
						{
							$\bar{u}_{\ast}^{2}  =u_{\ast}$, }
						$N_1,N_2\in[N_{\min};N_{\max}]$,
						and
						\begin{align}
						\!\begin{cases}\!
						V([\bar{f}(x_{N_1}^{1},\bar{u}_{\ast}^{1})]_{\mu_{\mathbf{X}}},x_{1}^{2}) \! \leq 
						(e^{\lambda \tau}\!+\!2) \gamma(\mu_{\mathbf{X}}),\\
						V([\bar{f}(x_{i}^{2},\bar{u}_{\ast}^{1})]_{\mu_{\mathbf{X}}},x_{i+1}^{2}) \! \leq 
						(e^{\lambda \tau}\!+\!2) \gamma(\mu_{\mathbf{X}})
						,   i \! \in\! [1; N_2-1];
						\end{cases}
						\label{eq:trans_rel_symb}
						\end{align}
						\item $Y_{\ast}=Y_{\tau}$;
						{\item $H_{\ast}\left({x}_0,\tilde{u}_0\right)=x_0$ for all ${x}_0\in [{\mathbb{R}^n}]_{\mu_{{\mathbf{X}}}}$;}
						\item  $H_{\ast}\left(x^{\ast}_{1},x^{\ast}_{2},...,x^{\ast}_{N},\bar{u}_{\ast}\right)=\left(x^{\ast}_{1},x^{\ast}_{2},...,x^{\ast}_{N}\right)$, for all $\left(x^{\ast}_{1},x^{\ast}_{2},...,x^{\ast}_{N},\bar{u}_{\ast}\right)\in X_{\ast}$, $N\in [N_{\min};N_{\max}]$.
					\end{itemize}
					\label{def:symb_model}
				\end{definition}
				
				System $S_{\ast}(\bar{\Sigma}_d)$ is pseudometric when $Y_{\ast}$ is equipped with the pseudometric $d_{Y_{\tau}}$. 
				%
				%
				%
				We can now present the following result. 
				
				\begin{theorem}
					\label{thmain}
					Consider $\bar{\Sigma}_d$ {in (\ref{sigma2}), with $N_k$ satisfying \eqref{eq:N_k},} and suppose that there exists a $\delta$-FC Lyapunov function $V$ for the control system $P$ in the NCS $\Sigma$. 
					Then, $S_{\ast}(\bar{\Sigma}_d) \preceq_{\varepsilon}^{s,\alt} S(\bar{\Sigma}_d)$ for any desired {accuracy} $\varepsilon\in\mathbb{R}^{+}$ and any state quantization $\mu_{{\mathbf{X}}}\in\mathbb{R}^{+}$ satisfying 
					\begin{equation}
					{
						\mu_{\mathbf{X}}=\hat{\mu}_\mathbf{X}/n_\mathbf{X}\leq \varepsilon,
					}
					\label{bisim_condition1}
					\end{equation}
					{
						for some integer $n_\mathbf{X}$.
					}
				\end{theorem}
				
				\begin{proof}
					Consider the relation $\mathcal{R}\subseteq X_{\ast}\times X_{\tau}$ defined by $\left(x^{\ast},x\right)\in\mathcal{R}$ if and only if 
					$x^{\ast}=(x^{\ast}_{1},x^{\ast}_{2},...,x^{\ast}_{N},\bar{u}_{\ast})$, $x=(x_{1},x_{2},...,x_{N},\bar{u})$, for some $N$, 
					$x_{i}^{\ast}=[x_{i}]_{\mu_{\mathbf{X}}}$, for all $i\in [1;N]$, and $\bar{u}_{\ast}=\bar{u}$.
					We first prove condition (i) of Definition \ref{ASR} in the Appendix. 
					{By definition of $[{\mathbb{R}^n}]_{\mu_{{\mathbf{X}}}}$, for any $x^{\ast}=(x^{\ast}_0,\tilde{u}_0)\in X_{0,\ast}$, 
						there exists $x=(x_0,\tilde{u}_0)\in X_{0,\tau}$ 
						with $x_0^{\ast}=[x_{0}]_{\mu_{\mathbf{X}}}$.} 
					We now consider condition (ii) of Definition \ref{ASR}. For any $(x^{\ast},x)\in\mathcal{R}$, from the definition of the pseudometric $d_{Y_{\tau}}$, the definition of $\mathcal{R}$ and condition (\ref{bisim_condition1}) we get $d_{Y_{\tau}}(H_{\ast}(x^{\ast}),H_{\tau}(x))  = \max_{i} \Vert x^{\ast}_{i}-x_{i} \Vert \leq  
					\mu_{\mathbf{X}} \leq \varepsilon$. We now show condition (iii$''$). Consider any $(x^{\ast},x)\in\mathcal{R}$, with $x^{\ast}=(x^{\ast}_{1},x^{\ast}_{2},...,x^{\ast}_{N},\bar{u}_{\ast})$ and $x=(x_{1},x_{2},...,x_{N},\bar{u})$; then pick any $u=u_{\ast}\in \mathbf{U}$ and consider any transition $x \rTo^u_{\tau} \bar{x}$, with $\bar{x}=(\bar{x}_1,\bar{x}_2,...,\bar{x}_{\bar{N}},u)$, for some $\bar{N}$. Pick $\bar{x}^{\ast}=(\bar{x}^{\ast}_1,\bar{x}^{\ast}_2,...,\bar{x}^{\ast}_{\bar{N}},u_{\ast})$ defined by $\bar{x}^{\ast}_i=[\bar{x}_i]_{\mu_{\mathbf{X}}}$ for all $i\in [1;\bar{N}]$. 
					By definition of $\bar{x}^{\ast}$ we get $(\bar{x}^{\ast},\bar{x})\in\mathcal{R}$. We conclude the proof by showing that 
					$x^{\ast} \rTo^{u_{\ast}}_{\ast} \bar{x}^{\ast}$, i.e. it is a transition of $S_{\ast}(\bar{\Sigma}_d)$. 
					By using condition (ii) in Definition \ref{TH-IFC}, one has $\frac{\partial{V}}{\partial{x_{N}^{\ast}}}(x_{N}^{\ast},x_{N}) f(x_{N}^{\ast},\bar{u}_{\ast}) + \frac{\partial{V}}{\partial{x_{N}}}(x_{N}^{\ast},x_{N}) f(x_{N},\bar{u}) \leq \lambda V(x_{N}^{\ast},x_{N})$. By the definitions of $\gamma$, $\mathcal{R}$ and $ S(\bar{\Sigma}_d)$, and by integrating the previous inequality, the following holds:
					\begin{align}
					V([\bar{f}(x^{\ast}_{N},\bar{u}_{\ast})]_{\mu_{\mathbf{X}}},\bar{x}^{\ast}_{1})  & \! \leq \! V(\bar{f}(x^{\ast}_{N},\bar{u}_{\ast}),\bar{x}^{\ast}_{1}) + \gamma(\mu_{\mathbf{X}}) \nonumber \\
					& \! \leq \! V(\bar{f}(x^{\ast}_{N},\bar{u}_{\ast}),\bar{x}_1) +2\gamma(\mu_{\mathbf{X}}) \nonumber \\
					& \! \leq \!  e^{\lambda \tau} V(x^{\ast}_N,x_N)+2\gamma(\mu_{\mathbf{X}}) \nonumber \\
					& \!\leq  \! e^{\lambda \tau} (V(x^{\ast}_N,[x_N]_{\mu_{\mathbf{X}}})\!+\!\gamma(\mu_{\mathbf{X}}))\!+\!2\gamma(\mu_{\mathbf{X}})  \nonumber \\
					& \! = \! (e^{\lambda \tau}+2) \gamma(\mu_{\mathbf{X}}),
					\label{eq_chain_of_ineq}
					\end{align}
					where the last equality holds by condition (i) of Definition \ref{TH-IFC}. By similar computations, it is possible to prove that 
					$V([\bar{f}(\bar{x}^{\ast}_{i},\bar{u}_{\ast})]_{\mu_{\mathbf{X}}},\bar{x}^{\ast}_{i+1}) \leq (e^{\lambda \tau}+2) \gamma(\mu_{\mathbf{X}}), \text{ } i\in [1;\bar{N}-1]$. 
					Hence, from the inequality above, from (\ref{eq_chain_of_ineq}) and from the definition of the transition relation of $S_{\ast}(\bar{\Sigma}_d)$ in (\ref{eq:trans_rel_symb}), we get $x^{\ast} \rTo^{u_{\ast}}_{\ast} \bar{x}^{\ast}$.
				\end{proof}

				\section{NCS Symbolic Control Design}\label{sec:control}
				
				In this section we provide the solution to Problem \ref{problem0}, which is based on the use of the symbolic models proposed in the previous section. 
				We first design a symbolic controller {system} $S_{C^*}$ that solves an appropriate approximate similarity game associated with Problem \ref{problem0}. We then refine the controller {system} $S_{C^*}$ to a controller $C^*$ in the form of (\ref{symbC}) {which solves} Problem \ref{problem0}. 

				We start by reformulating the specification $Q$ in Problem \ref{problem0} in terms of the following system:
				\begin{equation}
				S(Q)=(X_q,X^{0}_Q,U_q,\rTo_{q},Y_q,H_q),
				\label{ext_spec}
				\end{equation}
				where 
				\begin{itemize}
					\item $X_{q}= {X}_Q^0 \cup \{ x=(x_{1},x_{2},...,x_{N})  \in X_Q^N, N\in [N_{\min};N_{\max}] |  (x_{i}, x_{i+1}) \in T_Q, i\in {[1;N-1]} \}$;
					\item $U_q=\{ u_q \}$, where $u_q$ is a dummy symbol;
					\item $x^{1}\rTo_{q}^{u_q} x^{2}$, if $x^{1} =(x^{1}_{1},...,x^{1}_{N_{1}})$, $x^{2} =(x^{2}_{1},...,x^{2}_{N_{2}})$ and $x^1_{N_1} \rTo_{Q} x^2_{1}$; 
					\item $Y_q=Y_{\tau}$;
					\item $H_q(x)=x$, for all $x\in X_{q}$.
				\end{itemize}
				
				We now consider the following symbolic control problem:
				\begin{problem}
					\label{problem} 
					Consider the specification $S(Q)$ in (\ref{ext_spec}), the system $S(\bar{\Sigma}_d)$, and a desired {accuracy} $\varepsilon\in\mathbb{R}^{+}$. Find 
					a symbolic controller system $S_C$, some  parameters $\theta,\mu_{\mathbf{X}}\in\mathbb{R}^{+}$ and a strong $A\theta A$ simulation relation $\mathcal{R}$ from $S_C$ to $S(\bar{\Sigma}_d)$ {such that}:
					\begin{itemize}
						\item [1)] the $\theta$-approximate feedback composition of $S(\bar{\Sigma}_d)$ and $S_C$, denoted $S(\bar{\Sigma}_d)\times_{\theta}^\mathcal{R} S_C$,  
						is approximately simulated\footnote{The notions of approximate feedback composition and of approximate simulation are formally recalled in the Appendix.} by $S(Q)$ with accuracy $\varepsilon$, i.e. $S(\bar{\Sigma}_d)\times_{\theta}^\mathcal{R} S_C  \preceq_{\varepsilon} S(Q)$;
						\item [2)] the system $S(\bar{\Sigma}_d)\times_{\theta}^\mathcal{R} S_C$ is non-blocking; 
						\item [3)] for any pair of states $x=(x_{1},x_{2},...,x_{N},u)$ and $x'=(x'_{1},x'_{2},...,x'_{N},u')$ of $S(\bar{\Sigma}_d)$ if $[x_i]_{\mu_{\mathbf{X}}}=[x'_i]_{\mu_{\mathbf{X}}}$ for all $i\in [1;N]$, then $\mathcal{R}^{-1}(\{x\})=\mathcal{R}^{-1}(\{x'\})$.
					\end{itemize}
				\end{problem}
				
				The control design problem above, except for condition 3), is known in the literature as an approximate similarity game (see e.g. \cite{paulo}). 
				Condition 1) requires the state trajectories of the NCS to be close to the state run of the specification $S(Q)$ up to the accuracy $\varepsilon$ irrespective of the particular realization of the network non-idealities, and condition 2) prevents deadlocks in the interaction between the plant and the controller. 
				Condition 3) requires that aggregate states of $S(\bar{\Sigma}_d)$ with the same quantization are indistinguishable {for} the controller. By adding condition 3)  and by using the notion of strong alternating simulation relation (embedded in the notion of approximate feedback composition), we {are able to deal with} approximate similarity games where state measurements are only available through their quantizations. 
				Symbolic control problems for control systems with quantized state measurements and safety and reachability specifications have been studied in \cite{GirardNA2013}{. We also recall the recent work \cite{RWR15} that extends \cite{GirardNA2013} to general specifications for the class of nonlinear systems. The present control problem extends those considered in \cite{GirardNA2013} to NCS and specifications expressed as non-deterministic transition systems. }
				
				In order to solve Problem \ref{problem}, some preliminary definitions and results are needed. 
				Given two systems $S_{i}=(X_{i},X_{0,i},U_{i},$ $\rTo_{i},Y_{i},H_{i})$ ($i=1,2$), $S_{1}$ is a \textit{sub-system} of $S_{2}$ 
				if $X_{1}\subseteq X_{2}$, $X_{0,1}\subseteq X_{0,2}$, $U_{1}\subseteq U_{2}$, $\rTo_{1}\subseteq \rTo_{2}$, $Y_{1}\subseteq Y_{2}$ and $H_{1}(x)=H_{2}(x)$ for any $x\in X_{1}$. 
				Moreover, given two sub-systems $S_{i}=(X_{i},X_{0,i},U_{i},$ $\rTo_{i},Y_{i},H_{i})$ ($i=1,2$) of a system $S$, define the union system $S_{1} \bigsqcup S_2$ as $(X_{1} \cup X_{2},X_{0,1}\cup X_{0,2},U_{1}\cup U_{2}, \rTo_{1} \cup \rTo_{2},Y_{1} \cup Y_{2},H)$, where $H(x)=H_{1}(x)$ is $x \in X_1$ and $H(x)=H_{2}(x)$ otherwise. Note that $S_{1} \bigsqcup S_2$ is a sub-system of $S$. It is easy to see that the union operator enjoys the associative property.   
				We now have all the ingredients to introduce the controller $S_{C^{\ast}}$ that will solve Problem \ref{problem}. 
				
				
				\begin{definition}\label{canon_contr}
					The symbolic controller $S_{C^{\ast}}$ is the maximal non-blocking sub-system\footnote{Here maximality is defined with respect to the preorder induced by the notion of sub-system.} $S_C$ of $S{_{\ast}}(\bar{\Sigma}_d)$ such that: 
					\begin{itemize}
						\item [1)] $S_C$ is approximately simulated by $S(Q)$ with accuracy $\mu_{\mathbf{X}}$,  i.e. $S_C \preceq_{\mu_{\mathbf{X}}} S(Q)$;
						\item [2)] $S_C$ is {strongly} alternatingly $0$-simulated by $S{_{\ast}}(\bar{\Sigma}_d)$, i.e. \\ $S_C \preceq^{s,\alt}_{0}  S{_{\ast}}(\bar{\Sigma}_d)$.
					\end{itemize}
				\end{definition}
				
				Condition 1) of the definition above requires that for any state run $r_c$ of $S_C$ there exists a state run $r_q$ in  $S(Q)$  such that $r_c$ approximates $r_q$ within the accuracy $\mu_{\mathbf{X}}$.
				Condition 2) ensures that the controller enforces the specification irrespective of the time-delay realization induced  by the communication network. The following result holds.
				
				\begin{proposition}
					The symbolic controller $S_{C^{\ast}}$ is the union of all non-blocking sub-systems $S_{C}$ of $S{_{\ast}}(\bar{\Sigma}_d)$ satisfying conditions 1) and 2) of Definition  \ref{canon_contr}.
				\end{proposition}
				
				\begin{proof}
					Let $S_{C}$ and $S'_C$ be a pair of non-blocking sub-systems of $S_{\ast}(\bar{\Sigma}_d)$ satisfying both conditions 1) and 2) of Definition  \ref{canon_contr}. 
					Let $\mathcal{R}_a$ (resp. $\mathcal{R}'_a$) be a $\mu_{\mathbf{X}}$-approximate simulation relation from $S_{C}$ (resp. $S'_{C}$) to $S(Q)$. 
					Let $\mathcal{R}_b$ (resp. $\mathcal{R}'_b$) be a strong alternating $0$-approximate simulation relation from $S_{C}$ (resp. $S'_{C}$) to $S_{\ast}(\bar{\Sigma}_d)$. 
					Consider the system $S_{C}\bigsqcup S'_C$. By definition of operator $\bigsqcup$, relation $\mathcal{R}_a \cup \mathcal{R}'_a$ is a $\mu_{\mathbf{X}}$-approximate simulation from $S_{C} \bigsqcup S'_{C}$ to $S(Q)$, and relation $\mathcal{R}_b \cup \mathcal{R}'_b$ is a strong alternating $0$-approximate simulation from $S_{C} \bigsqcup S'_{C}$ to $S_{\ast}(\bar{\Sigma}_d)$. Hence, $S_{C} \bigsqcup S'_{C}$ satisfies condition 1) and 2) of Definition \ref{canon_contr}. Moreover, since $S_{C}$ and $S'_C$ are non-blocking, again by definition of operator $\bigsqcup$, system $S_{C} \bigsqcup S'_{C}$ is non-blocking as well. Finally, since 
					$S_{C^{\ast}}$ is the union of all non-blocking sub-systems $S_{C}$ of $S_{\ast}(\bar{\Sigma}_d)$ satisfying conditions 1) and 2) of Definition \ref{canon_contr}, it is the maximal non-blocking sub-system $S_C$ of $S_{\ast}(\bar{\Sigma}_d)$ satisfying conditions 1) and 2) of Definition \ref{canon_contr}.
				\end{proof}
				
				{Although $S{_{\ast}}(\bar{\Sigma}_d)$ is countable, since the set $\mathbf{X}$ is bounded and $S(Q)$ is symbolic, the controller system $S_{C^{\ast}}$ is symbolic and} can be computed in a finite number of steps  
				by adapting standard fixed point characterizations of {simulation}  \cite{ModelChecking,paulo}. 
				We  now  provide the solution {to} Problem \ref{problem}. 
				
				\begin{theorem}
					Consider the NCS $\Sigma$ and the specification $S(Q)$. Suppose that there exists a $\delta$-FC Lyapunov function $V$ for the control system $P$  in the NCS $\Sigma$.   
					For any desired {accuracy} $\varepsilon\in\mathbb{R}^{+}$, choose the parameters $\theta,\mu_{{\mathbf{X}}}
					\in\mathbb{R}^{+}$ such that:
					\begin{equation}
					{
						\mu_{{\mathbf{X}}}+\theta  \leq\varepsilon\label{solving_condition_1}}
					\end{equation}
					
					{ with $\mu_{\mathbf{X}}=\hat{\mu}_\mathbf{X}/n_\mathbf{X}$, for some integer $n_\mathbf{X}$.} 
					Then a strong $A \theta A$ simulation relation $\mathcal{R}$ from $S_{C^{\ast}}$ to $S(\bar{\Sigma}_d)$ exists solving Problem \ref{problem} with $S_C=S_{C^{\ast}}$.
					\label{Mmain}
				\end{theorem}
				
				\begin{proof}
					By condition 2) in Definition \ref{canon_contr}, a (non-empty) strong $A0A$ simulation relation $\mathcal{R}_1$ from $S_{C^*}$ to $S_{\ast}(\bar\Sigma_d)$ exists. Let $\mathcal{R}_2$ be the relation defined in the proof of Theorem \ref{thmain}. 
					Since there exists a $\delta$-FC Lyapunov function for the plant $P$ and condition (\ref{solving_condition_1}) holds, by Theorem {\ref{thmain}}, $\mathcal{R}_2$ is a strong $A\theta A$ simulation relation from $S_{\ast}(\bar\Sigma_d)$ to $S(\bar\Sigma_d)$. 
					Define the relation $\mathcal{R}=\mathcal{R}_1 \circ \mathcal{R}_2$. By Lemma 1 (ii), $\mathcal{R}$ is a strong $A\theta A$ simulation relation from $S_{C^*}$ to $S(\bar\Sigma_d)$. 
					We start by showing condition 1) of Problem \ref{problem}. 
					The existence of $\mathcal{R}_1$ and $\mathcal{R}_2$ implies by Definition \ref{ASR} that $S_{C^{\ast}} \preceq^{s,\alt}_{0} S_{\ast}(\bar{\Sigma}_d)$ and $S_{\ast}(\bar{\Sigma}_d) \preceq_{\theta}^{s,\alt} S(\bar{\Sigma}_d)$. Hence, from Lemma \ref{lemma1} (ii) in the Appendix, by combining the previous implications, one gets $S_{C^{\ast}}\preceq_{\theta}^{s,\alt} S(\bar{\Sigma}_d)$ which, by Lemma \ref{lemma1} (iii), leads to ${S(\bar{\Sigma}_d)\times^{\mathcal{R}}_{\theta} S_{C^{\ast}}} \preceq_{\theta} S_{C^{\ast}}$. 
					Since $S_{C^{\ast}} \preceq_{\mu_{\mathbf{X}}} S(Q)$ by condition 1) in Definition \ref{canon_contr}, Lemma \ref{lemma1} (ii) and condition (\ref{solving_condition_1})  imply $S(\bar{\Sigma}_d)\times^{\mathcal{R}}_{\theta} S_{C^{\ast}} \preceq_{\varepsilon} S(Q)$.  
					We now show that condition 2) holds. Consider any state $(x,x_c)$ of $S(\bar{\Sigma}_d)\times^{{\mathcal{R}}}_{\theta} S_{C^{\ast}}$. Pick any $u_{c}\in U_{c}(x_c)$, which is a non-empty set because $S_{C^{\ast}}$ is non-blocking. Since $(x_c,x)\in \mathcal{R}$,  
					for any $x\rTo_{\tau}^{u} x^{\prime}$ in $S(\bar{\Sigma}_d)$ there exists $x_c\rTo_{c}^{u}x^{\prime}_{c}$ in $S_{C^{\ast}}$ with $(x^{\prime}_{c},x^{\prime})\in \mathcal{R}$. Hence, from Definition \ref{composition}, the transition $(x,x_{c}) \rTo^u (x^{\prime},x^{\prime}_{c})$ is in $S(\bar{\Sigma}_d)\times^{{\mathcal{R}}}_{\theta} S_{C^{\ast}}$, implying that $S(\bar{\Sigma}_d)\times^{\mathcal{R}}_{\theta} S_{C^{\ast}}$ is non-blocking. 
					We conclude by showing condition 3). 
					Consider a pair of states $x=(x_{1},x_{2},...,x_{N},u)$ and $x'=(x'_{1},x'_{2},...,x'_{N},u')$ of $S(\bar{\Sigma}_d)$ such that $[x_i]_{\mu_{\mathbf{X}}}=[x'_{i}]_{\mu_{\mathbf{X}}}$ for all $i\in [1;N]$. Since $\mathcal{R}_2^{-1}(\{x\})=\{[x]_{\mu_{\mathbf{X}}}\}$, $\mathcal{R}_2^{-1}(\{x'\})=\{[x']_{\mu_{\mathbf{X}}}\}=\{[x]_{\mu_{\mathbf{X}}}\}$, by recalling that  $\mathcal{R}^{-1}(\{x\})=\mathcal{R}_1^{-1}(\mathcal{R}_2^{-1}(\{x\}))$ and 
					$\mathcal{R}^{-1}(\{x'\})=\mathcal{R}_1^{-1}(\mathcal{R}_2^{-1}(\{x'\}))$, we get condition 3). 
				\end{proof}
				
				
				We now proceed with a further step by refining the controller $S_{C^\ast}$ solving Problem \ref{problem} to a controller {$C^{\ast}$} in form of (\ref{symbC}) which can be applied to the original NCS and {solves} Problem \ref{problem0}.
				Let {$U_{C^{\ast}}(\cdot)$} and $\Post (\cdot)$ be the operators defined in Definition \ref{ssystem} but applied to system $S_{C^{\ast}}$.  
				Let  $S_{C^{\ast}}=(X_{C^{\ast}},X_{0,C^{\ast}},U_{C^{\ast}},\rTo_{C^{\ast}},$ $Y_{C^{\ast}},H_{C^{\ast}})$. Define $\Xi_C={X_{C^{\ast}}}$,  $\Xi^0_C=  X_{0,C^{\ast}}$ and
				{
					\begin{equation}
					\begin{cases}
					h_C(\xi)  \in U_{C^{\ast}}(\xi), \\ 
					f_C(\xi,w)  = \{\xi'=\left(\xi'_1,...,\xi'_{N'},\bar{u}\right) \in \Post_{h_C(\xi)} (\xi):\xi'_{N'}=w\},
					\end{cases}
					\label{eq:controller_fun}
					\end{equation}
					for any $(\xi,w) \in 
					\Xi_C \times [\mathbf{X}]_{\mu_{{\mathbf{X}}}}$.}  
				Note from the first line in (\ref{eq:controller_fun}) that the controller $C$, as in (\ref{symbC}), derived from a non-blocking non-deterministic system $S_{C^{\ast}}$ is not uniquely determined, since {$U_{C^{\ast}}(\xi)\ne \varnothing$} may  not be a singleton. Moreover, 
				the second line in (\ref{eq:controller_fun})} takes into account that {$\xi'_{N'}$} is the state of the aggregate vector $x^{\ast}$ in {$\xi'$} which is required to match the output sample $w$, sent through the plant-to-controller branch of the network and reaching the controller (as illustrated in Section \ref{sec:modelingNCS}).  We conclude this section by proving the formal correctness of the controller $C^{\ast}$ as defined above.
			
			\begin{theorem}
				Assume that the conditions of Theorem \ref{Mmain} hold, implying the existence of some parameters $\theta,\mu_{\mathbf{X}}\in\mathbb{R}^{+}$ {satisfying the inequality in (\ref{solving_condition_1}),}
				{ with $\mu_{\mathbf{X}}=\hat{\mu}_\mathbf{X}/n_\mathbf{X}$ for some integer $n_\mathbf{X}$}, of a symbolic controller system $S_C=S_{C^{\ast}}$ and of a strong $A\theta A$ simulation relation $\mathcal{R}$ from $S_C$ to $S(\bar{\Sigma}_d)$ solving Problem \ref{problem}. Set $\mathbf{X}_0$ such that $\mathcal{R}(X_{0,C^{\ast}})={{\mathbf{X}}_0}  \times \{\tilde{u}_0\}$. Then the controller $C^{\ast}$ solves Problem \ref{problem0}.
			\end{theorem}
			
			\begin{proof}
				Consider the strong $A\theta A$ simulation relation $\mathcal{R}=\mathcal{R}_1 \circ \mathcal{R}_2$ from $S_{C^{\ast}}$ to $S(\bar{\Sigma}_d)$ defined in the proof of {Theorem {\ref{Mmain}}. 
					Now consider any $\tilde{y}_0={x(0)}\in \mathbf{X}_0$,  implying that $x^0=({x(0)},\tilde{u}_0)\in\mathcal{R}(X_{0,C^{\ast}})$ by definition of  ${{\mathbf{X}}_0}$. Then consider the state $\xi_1:=([{x(0)}]_{\mu_{\mathbf{X}}},\tilde{u}_0)$; by definition of $\mathcal{R}$ we get $\xi_1\in\mathcal{R}^{-1}(x^0)$, 
					implying that $\xi_1 \in X_{0,C^{\ast}}$. From the first line in the refinement equation (\ref{eq:controller_fun}), the control input $v_1= h_C(\xi_1)\in U_{C^{\ast}}(\xi_1)$ is uniquely determined. 
					Furthermore, since  $(\xi_1,x^0)\in\mathcal{R}$, which is a strong $A\theta A$ simulation relation  from $S_{C^{\ast}}$ to $S(\bar{\Sigma}_d)$, then $v_1 \in {\mathbf{U}}(x^0)$ in $S(\bar{\Sigma}_d)$ and, for any transition $x^0 \rTo^{v_{1}} x^1=(\bar{x}^1,v_1)=((\bar{x}^{1}_{1}, {...}, \bar{x}^{1}_{N_{1}}),v_1)$ in $S(\bar{\Sigma}_d)$, there exists  a transition $\xi_1 \rTo^{v_{1}}  \xi_2=(({\xi}_{2,1}, {...}, {\xi}_{2,N_{1}}),v_1)$ in $S_{C^*}$ such that $(\xi_2,x^1)\in \mathcal{R}$, implying $\xi_{2,N_1}=[x^1_{N_1}]_{\mu_{\mathbf{X}}}$ from the definition of $\mathcal{R}$. 
					By induction, assume now $(\xi_{k},x^{k-1}) \in \mathcal{R}$ for some $k\in \mathbb{N}$, with $x^{k-1}$ in the form $x^{k-1}=(\bar{x}^{k-1},v_{k-1})$, and again by  exploiting the non-blocking property of $S_{C^{\ast}}$, the definition of $\mathcal{R}$ and the refinement equation (\ref{eq:controller_fun}), it is readily seen that by choosing $v_{k}= h_C(\xi_{k}) \in U_{C^{\ast}}(\xi_{k})$, then one has $v_{k} \in {\mathbf{U}}(x^{k-1})$ in $S(\bar{\Sigma}_d)$ and, for any transition $x^{k-1} \rTo^{v_{k}} x^{k}=(\bar{x}^{k},v_{k})=((\bar{x}^{k}_{1}, {...}, \bar{x}^{k}_{N_{k}}),v_{k})$ in $S(\bar{\Sigma}_d)$, there exists  a transition $\xi_{k} \rTo^{v_{k}}  \xi_{k+1}=(({\xi}_{k+1,1}, {...}, {\xi}_{k+1,N_{k}}),v_{k})$ in $S_{C^*}$ such that $(\xi_{k+1},x^{k})\in \mathcal{R}$, implying $\xi_{k+1,N_{k}}=[x^{k}_{N_{k}}]_{\mu_{\mathbf{X}}}$ from the definition of $\mathcal{R}$. 
					As a result of the procedure above, we built an infinite sequence $\{(\xi_{k},x^{k-1})\}_{k\in \mathbb{N}}\subseteq \mathcal{R}$ and two infinite state runs $\xi_1 \rTo^{v_{1}}  \xi_2 \rTo^{v_{2}}\xi_3 \rTo^{v_{3}}  \,\, {...}\,$ and  $x^0 \rTo^{v_{1}} x^1 \rTo^{v_{2}} x^2 \rTo^{v_{3}}  \,\, {...}\,$ in $S_{C^*}$ and $S(\bar{\Sigma}_d)$, respectively. By Definition \ref{composition} of approximate feedback composition, this implies that  
					\begin{equation}
					(x^0,\xi_1) \rTo^{{v}_{1}}  (x^1,\xi_2) \rTo^{{v}_{2}}(x^2,\xi_3) \rTo^{{v}_{3}}  \,\, {...}\,
					\label{cond_2_tris}
					\end{equation}
					is an infinite state run of $S(\bar{\Sigma}_d)\times_{\theta}^\mathcal{R} S_{C^*}$.
					From {Proposition} \ref{theo:theo1}, the existence of an infinite state run  $x^0 \rTo^{v_{1}} x^1 \rTo^{v_{2}} x^2 \rTo^{v_{3}}  \,\, {...}\,$ in $S(\bar{\Sigma}_d)$ implies the existence of an infinite trajectory $\{\tilde{y}_s \}_{{s \in \mathbb{N}_0}}=\{z_s \}_{s \in \mathbb{N}_0}$ of 
					$\Sigma_d$
					such that 
					\begin{align}
					\{ {{x(0)}}  , \quad \underbrace{\bar{x}^{1}_{1}, {...}, \bar{x}^{1}_{N_{1}}}_{\bar{x}^{1}}  , \quad \underbrace{\bar{x}^{2}_{1}, {...}, \bar{x}^{2}_{N_{2}}}_{\bar{x}^{2}}  , {...} \} = \{z_s \}_{s \in \mathbb{N}_0}= \{\tilde{y}_s \}_{s \in \mathbb{N}_0}.
					\label{cond2_bis}
					\end{align}
					From the definition of quantizer and switch in \eqref{sigma2}, one can write, for any $k\in \mathbb{N} \setminus\{1\}$,  $w_{k}=y_{M_{k}}=[\tilde{y}_{M_{k}}]_{\mu_{\mathbf{X}}}=[x^{k-1}_{N_{k-1}}]_{\mu_{\mathbf{X}}}=\xi_{{k},N_{k-1}}$. This implies, from the second line in (\ref{eq:controller_fun}), that $\xi_{k} \in f_C(\xi_{k-1},w_{k})$, so the evolution of the controller in (\ref{symbC}) is well defined at all iterations $k$. Finally, from {Proposition} \ref{TeoEquivalence}, the existence of the trajectory $\{z_s \}_{s \in \mathbb{N}_0}$ {of  $\Sigma_d$} in \eqref{cond2_bis} implies that there exists a trajectory $x:[0,+\infty[\rightarrow \mathbb{R}^n$ of the NCS $\Sigma$ such that $\tilde{y}_s=z_s=x(\tau s)$ for all $s \in \mathbb{N}_0$. This concludes the proof that any sequence $\{\tilde{y}_s\}$ generated by the NCS is defined for all $s \in \mathbb{N}_0$.
				}
				Since the assumptions of Theorem \ref{Mmain} hold,  condition 1) of Problem \ref{problem} is fulfilled by the controller $S_{C^*}$, i.e. $S(\bar{\Sigma}_d)\times_{\theta}^\mathcal{R} S_{C^*}  \preceq_{\varepsilon} S(Q)$. Hence,  Definition \ref{ASR} (approximate simulation) implies that, for any initial state $(x^0,\xi_1)$ of $S(\bar{\Sigma}_d)\times_{\theta}^\mathcal{R} S_{C^*} $, there exists $x_q^0 \in X_{Q}^{0}$ such that $d_{Y_{\tau}}(H_{\tau}(x^0),H_q(x_q^0))  = \Vert {x(0)}-x_q^0 \Vert  \leq   \varepsilon$, and  the existence of a state run (\ref{cond_2_tris}) in $S(\bar{\Sigma}_d)\times_{\theta}^\mathcal{R} S_{C^*} $ implies the existence of a state run 
				\begin{equation}
				x_q^0 \rTo^{u_{q}}_q  x_q^1 \rTo^{u_{q}}_q x_q^2 \rTo^{u_{q}}_q  \,\, {...}\,
				\label{eq_state_run_q}
				\end{equation}
				in $S(Q)$, with $x_q^k$ in the form $x^k_q=(x^k_{q,1},...,x^k_{q,N_k})$, such that $d_{Y_{\tau}}(H_{\tau}(x^k),H_q(x_q^k))  = \max_{i} \Vert \bar{x}^{k}_{i}-x^k_{q,i} \Vert  \leq \varepsilon$, implying
				\begin{equation}
				\Vert \bar{x}^k_i - x^k_{q,i} \Vert \leq \varepsilon, \quad \forall k \in {\mathbb{N}} \text{ and }  \forall i=1,...,N_k.
				\label{component_epsilon}
				\end{equation}
				In turn, from the definition of specification $Q$,  the existence of a state run in $S(Q)$ in Eq. (\ref{eq_state_run_q}) implies the existence in $Q$ of the transitions $(x^s_Q,x^{s+1}_Q)\in T_Q$, for all $s\in \mathbb{N}_0$, such that: 
				\begin{align}
				\{ {{x}_q^{0}} \quad ,  \underbrace{{x}^{1}_{q,1}, {...}, {x}^{1}_{q,N_{1}}}_{{x}_q^{1}}  , \quad \underbrace{{x}^{2}_{q,1}, {...}, {x}^{2}_{q,N_{2}}}_{{x}_q^{2}}  ,  {...} \} = \{x^s_Q \}_{s \in \mathbb{N}_0}.
				\label{cond3_bis}
				\end{align}
				Hence, condition 1)  of Problem \ref{problem0} holds. Finally, by (\ref{cond2_bis}), (\ref{cond3_bis}), and (\ref{component_epsilon}), we get condition 2) of Problem \ref{problem0}. 
			\end{proof}
			
			\section{Application to Robot Motion Planning \\ with Remote Control}\label{sec:example}
			Symbolic techniques for robot motion planning and control have been {successfully}  exploited in the literature, see e.g. \cite{BeltaIEEERob2007} and the references therein. However, existing work does not consider the symbolic control of robot motion over non-ideal communication networks. In this section we exploit the remote control of an electric car-like  robot, with limited power, sensing, computation and communication capabilities, whose goal is the surveillance of an area. 
			The motion of the robot $P$ is described by means of the following nonlinear control system:
			
			\begin{equation}
			\left[
			\begin{array}
			[l]{l}
			\dot{x}_1\\
			\dot{x}_2\\
			\dot{x}_3
			\end{array}
			\right] =
			\left[
			\begin{array}
			[c]{c}
			u_1 \frac{\cos(x_3+\delta(u_2))}{\cos(\delta(u_2))}\\
			u_1 \frac{\sin(x_3+\delta(u_2))}{\cos(\delta(u_2))}\\
			\frac{u_1}{b}\tan(u_2)
			\end{array}
			\right],
			\label{example1}
			\end{equation}
			
			where $\delta(u_2)=\arctan\left(\frac{a\tan(u_2)}{b}\right)$, $a=0.5$ is the distance of the center of mass from the rear axle and $b=1.5$ is the wheel base, see Fig. \ref{fig:sim} (left panel) {(modified from Fig. 2.16 in \cite{astrom2010feedback})}. 
			States $x_1$ and $x_2$ are the $2$D-coordinates of the center of mass of the vehicle and state  $x_3$ is  its heading angle, while the inputs  $u_1$ and $u_2$ are the velocity of the rear wheel and the steering angle,  respectively. Note that $u_1$ is always nonnegative to guarantee that the vehicle does not move backwards. All the quantities are expressed in units of the International System (SI). {We consider an {accuracy} $\varepsilon=0.02$, and the bounded set including all the specification trajectories up to $\varepsilon$ is}
			$\mathbf{X}={\left[-x_{1,\max},x_{1,\max}\right[\times \left[-x_{2,\max},x_{2,\max}\right[ \times \left[-x_{3,\max},x_{3,\max}\right[}$, and $u\in \mathbf{U} {\subset\left[0,u_{1,\max}\right[\times \left[-u_{2,\max},u_{2,\max}\right[}$, 
			where $x_{\max}=[x_{1,\max},x_{2,\max},x_{3,\max}]'=[50,50,\pi]'$ 
			and $u_{\max}=[u_{1,\max},u_{2,\max}]'=[5,\frac{\pi}{3}]'$. 
			The model above is known in the literature as \emph{single-track} vehicle model and is widely used because, in spite of its simplicity, it well captures the major features of interest of the vehicle cornering behavior \cite{gillespie1992vehicle}. 
			The robot $P$ is remotely connected to a controller, implemented on a shared CPU, by means of a non-ideal communication network. The control loop forms a NCS, as the one in Fig. 1, whose network/computation parameters are $B_{\min} = 0.1 \text{ kbit}/s$, $B_{\max} =1 \text{ kbit}/s$, $\tau =1 s$, $\Delta_{\min}^{\ctrl}=0.01 s$, $\Delta_{\max}^{\ctrl} =0.1 s$, $\Delta^{\req}_{\min} =0.05 s$, $\Delta^{\req}_{\max} =0.2 s$, $\Delta_{\min}^{\net}=0.1 s$, $\Delta_{\max}^{\net}=0.25 s$. {Given the different nature of the three state variables,} the state quantization is assumed to be different (in absolute values) for each component and equal to $x_{i,\max}/100$ for the state $x_i$ ($i=1,2,3$), so that we have {$200$} quantization values for each state component. {Similarly,} we assume the input quantization to be equal to $u_{i,\max}/5$ for the input $u_i$ ($i=1,2$) and the network protocols to introduce a relative overhead which is bounded by the $20 \%$ of the total number of data bits ($N_{\cp}^{+}=N_{\pc}^{+}=0.2$). This implies $ \vert [\mathbf{X}]_{\mu_{{\mathbf{X}}}} \vert ={200^3}$ and $\vert \mathbf{U} \vert ={50}$, hence  {$\Delta_{{\min}}^{B,\pc}=0.0276 s$, $\Delta_{{\max}}^{B,\pc}=0.276 s$, $\Delta_{{\min}}^{B,\cp}=0.0072 s$, $\Delta_{{\max}}^{B,\cp}=0.072 s$}. We assume there may be packet dropouts, with the constraint that two consecutive dropouts are not allowed ($N_{\pd}=1$). 
			The motion planning problem considered here is described in the following. We require that the robot leaves its support (HOME location) and visits (in the exact order) two buildings, denoted by $B1$ and $B2$, to then reach an outlet where it possibly powers up the battery (CHARGE location). Finally, the vehicle returns HOME. During the whole path, the robot is requested to avoid some obstacles, such as walls and other buildings. We denote the union of the obstacles locations as the UNSAFE location. 
			We now start applying the results in Section \ref{sec:SymbolicModels} regarding the design of a symbolic model for the given NCS. According to the definition of $\Sigma_d$, 
			the minimum and maximum delays in a single iteration of the network amount to {$\Delta_{\min}=0.24 s$  and $\Delta_{\max}=2.07 s$}, respectively. From (\ref{eq:N_min_max}), this results in $N_{\min}=1$, $N_{\max}=3$. 
			In order to have a uniform quantization in the state space, we apply the results to a normalized plant $\tilde{P}$, whose state  is the one of $P$, but component-wise normalized with respect to $x_{\max}$. According to the previous {description} of the NCS, this results in {$\hat{\mu}_{\mathbf{X}}=1$, $n_{\mathbf{X}}=200$ and $\mu_{\mathbf{X}}=0.005$}. We assume that the normalized signals are sent through the network and the static block implementing the coordinate change from $P$ to $\tilde{P}$ and vice versa (omitted in the general scheme) {is} physically connected to the sensor. 
			It is possible to show that the quadratic Lyapunov-like function $V(x,x^{\prime})=0.5\, \Vert x-x^{\prime}\Vert_{2}^{2}$, is $\delta$-FC for 
			control system (\ref{example1}), with $\lambda=\frac{2u_{1,\max}}{\cos(\delta(u_{2,\max}))}$, $\underline{\alpha}(r)=0.5 r^2$, $\overline{\alpha}(r)=1.5 r^2$ and $\gamma(r)=6 r$; hence Theorem \ref{thmain} can be applied. 
			In the symbolic control design step, we apply the results illustrated in Section \ref{sec:control}. We first construct a finite transition system $Q$ which encodes a number of randomly generated trajectories satisfying the given specification. 
			For the choice of $\theta =0.0125$, Theorem \ref{Mmain} holds and the controller $S_{C^{\ast}}$ in Definition \ref{canon_contr} solves the control problem. 
			Estimates of the space complexity in constructing $S_{C^{\ast}}$ indicate $4\cdot10^{13}$ 32-bit integers. 
			Because of the large computational complexity in building the controller $S_{C^{\ast}}$, we do not construct the
			whole symbolic model $S_{\ast}(\bar{\Sigma}_d)$, from which deriving $S_{C^{\ast}}$, but only the part of $S_{\ast}(\bar{\Sigma}_d)$ that can implement (part of) the specification $Q$; similar ideas were explored in \cite{BorriCDC2012}, see also \cite{PolaTAC12}.
			The total memory occupation and time required to construct  $S_{C^{\ast}}$ are respectively $3742$ $32$-bit integers and $2833\,$s. 
			The computation has been performed on the Matlab suite through an Apple MacBook Pro with 2.5GHz Intel Core i5 CPU and 16 GB RAM. 
			In Fig. \ref{fig:sim} (right panel)
			, we show a sample path of the NCS (blue solid line), for a particular realization of the network uncertainties, compared to the trajectory of the system controlled through an ideal network (black dash-dotted line). 
			Each time delay realization $N_k$ is sampled from a discrete uniform random distribution over $[N_{\min};N_{\max}]$. 
			As a result, the NCS used just $59$ control samples, in spite of the $94$ control samples (one at each $\tau$) used in the ideal case. 
			Note that, although the behavior of the NCS is not as regular as in the ideal case, the specification is indeed met.
			
			\begin{figure*}
				\centering
				\subfigure
				{\includegraphics[scale=0.4]{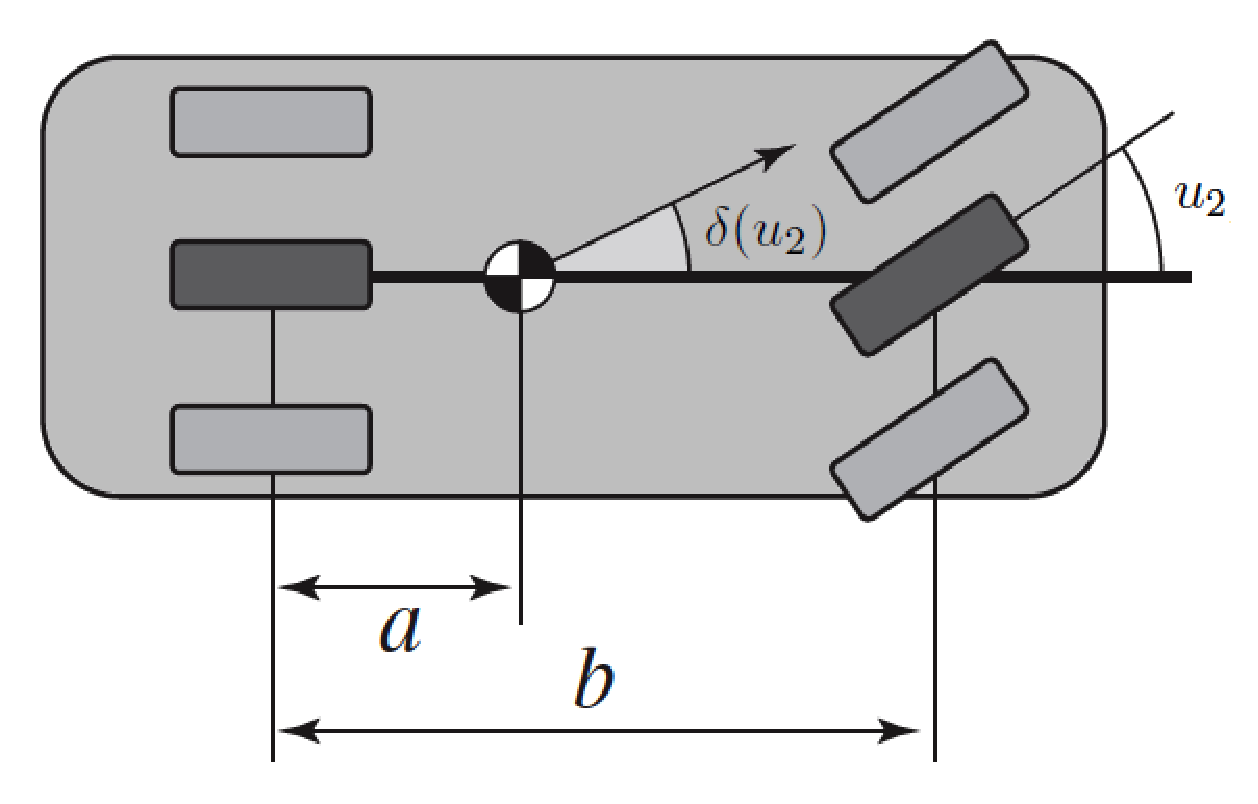}}
				\hspace{0.5mm}
				\subfigure
				{\includegraphics[scale=0.44]{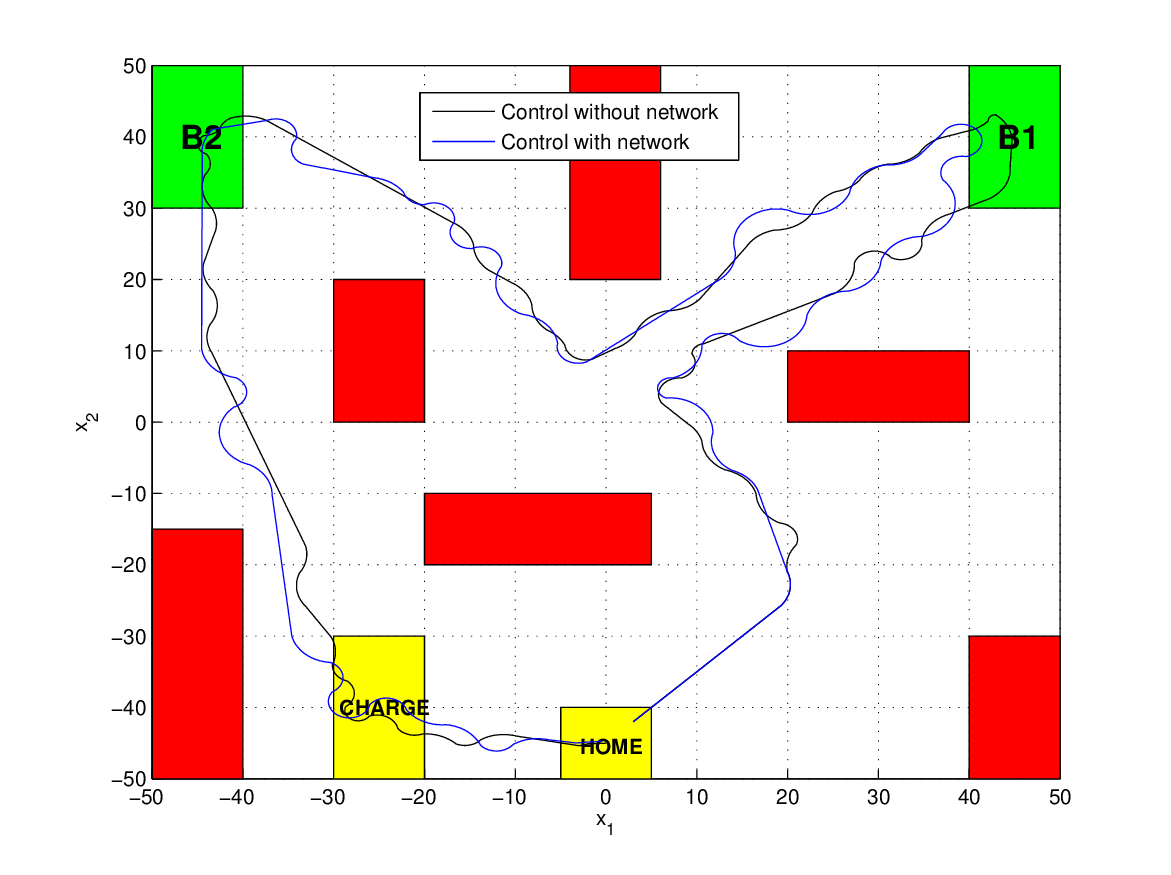}}
				\caption{Overhead view of the robot dynamics (top panel). Space trajectory of the vehicle (bottom panel). }
				\label{fig:sim}
			\end{figure*}
			
			\section{Conclusions}\label{sec:conclusion}
			In this paper we proposed a symbolic approach to the control design of nonlinear NCS, where the most important non-idealities in the communication channel are taken into account. Under the assumption of existence of incremental forward complete Lyapunov functions, we derived symbolic models that approximate NCS in the sense of strong alternating approximate simulation. 
			NCS symbolic control design, where specifications are expressed in terms of transition systems, was then solved and applied to an example 
			{of remote robot motion planning}.
			
			
			\section*{Acknowledgements}
			The authors are grateful to Pierdomenico Pepe for fruitful discussions on the
			topic of this article.
			
			\bibliographystyle{abbrv}
			\bibliography{biblio1}
			
			\appendix
			We here recall from \cite{AB-TAC07,PolaSIAM2009}, {the notion} of (alternating) approximate {simulation} relations and introduce the  {notion} of strong alternating { approximate simulation} relations. Approximate feedback composition is also introduced and adapted from \cite{paulo}. 
			
			\begin{definition}
				\label{ASR}
				Let $S_{i}=(X_{i},X_{0,i},U_{i},\rTo_{i},$ $Y_{i},H_{i})$ ($i=1,2$) be (pseudo)metric systems with the same output sets $Y_{1}=Y_{2}$ and (pseudo)metric $d$, and let $\varepsilon\in\mathbb{R}^{+}_{0}$ be a given {accuracy}. Consider a relation $\mathcal{R}\subseteq X_{1}\times X_{2}$ satisfying the following conditions:
				\begin{itemize}
					\item
					[(i)] $\forall x_{1}\in X_{0,1}$ $\exists x_{2}\in X_{0,2}$ such that $(x_{1},x_{2})\in \mathcal{R}$;
					\item [(ii)] $\forall (x_{1},x_{2})\in \mathcal{R}$, \mbox{$d(H_{1}(x_{1}),H_{2}(x_{2}))\leq\varepsilon$}.
				\end{itemize}
				Relation $\mathcal{R}$ is an \emph{$\varepsilon$-approximate simulation relation} from $S_{1}$ to $S_{2}$ if it enjoys conditions (i), (ii) and the following one: 
				\begin{itemize}
					\item[(iii)] $\forall (x_{1},x_{2})\in \mathcal{R}$ if \mbox{$x_{1}\rTo_{1}^{u_{1}}x'_{1}$} then \mbox{$\exists  x_{2}\rTo_{2}^{u_{2}}x'_{2}$} such that $(x^{\prime}_{1},x^{\prime}_{2})\in \mathcal{R}$.
				\end{itemize}
				System $S_{1}$ is $\varepsilon$-simulated by $S_{2}$ or $S_{2}$ $\varepsilon$-simulates $S_{1}$, denoted \mbox{$S_{1}\preceq_{\varepsilon}S_{2}$}, if there exists an $\varepsilon$-approximate simulation relation from $S_{1}$ to $S_{2}$.
				Relation $\mathcal{R}$ is an \emph{alternating $\varepsilon$-approximate ($A\varepsilon A$) simulation relation} from $S_{1}$ to $S_{2}$ if it enjoys conditions (i), (ii) and the following one: 
				\begin{itemize}
					\item
					[(iii$'$)] $\forall (x_{1},x_{2})\in\mathcal{R}$ $\forall u_{1}\in U_{1}(x_1)$ $\exists u_{2}\in U_{2}(x_2)$ such that \mbox{$\forall x_{2}\rTo_{2}^{u_{2}}x'_{2}$} \mbox{$\exists x_{1}\rTo_{1}^{u_{1}}x'_{1}$} with $(x_{1}^{\prime},x_{2}^{\prime} )\in\mathcal{R}$. 
				\end{itemize}
				Relation $\mathcal{R}$ is a \emph{strong alternating $\varepsilon$-approximate (strong $A\varepsilon A$) simulation relation} from $S_{1}$ to $S_{2}$ if it enjoys conditions (i), (ii) and the following one: 
				\begin{itemize}
					\item [(iii$''$)] $\forall (x_{1},x_{2})\in\mathcal{R}$ $\forall u_{1}\in U_{1}(x_1)$, $u_{2}=u_1\in U_{2}(x_2)$ and \mbox{$\forall x_{2}\rTo_{2}^{u_{2}}x'_{2}$} \mbox{$\exists x_{1}\rTo_{1}^{u_{1}}x'_{1}$} such that $(x_{1}^{\prime},x_{2}^{\prime} )\in\mathcal{R}$. 
				\end{itemize}
				System $S_{1}$ is strongly alternatingly $\varepsilon$-simulated by $S_{2}$ or $S_{2}$ strongly alternatingly $\varepsilon$-simulates $S_{1}$, denoted \mbox{$S_{1}\preceq_{\varepsilon}^{s,\alt} S_{2}$}, if there exists a strong $A\varepsilon A$ simulation relation from $S_{1}$ to $S_{2}$. 
			\end{definition}
			
			The  notion of strong $A \varepsilon A$ {simulation} relation has been inspired by the notion of feedback refinement relations recently introduced in \cite{RWR15}. 
			Interaction between plants and controllers in the systems domain is formalized as follows:
			\begin{definition}
				\cite{paulo}
				\label{composition}
				Consider a pair of (pseudo)metric systems $S_{i}=(X_{i},X_{0,i},U_{i},\rTo_{i},$ $Y_{i},H_{i})$ ($i=1,2$) with the same output sets
				$Y_{1}=Y_{2}$ and (pseudo)metric $d$, and let $\varepsilon\in\mathbb{R}^{+}_{0}$ be a given {accuracy}. Let $\mathcal{R}$ be a strong $A \varepsilon A$ simulation relation from $S_{2}$ to $S_{1}$.
				The $\varepsilon$-approximate feedback composition of $S_{1}$ and $S_{2}$, with composition relation $\mathcal{R}$, is the system 
				$S_{1}\times^{\mathcal{R}}_{\varepsilon}S_{2}=(X,X_{0},U,\rTo,Y,H)$, where 
				$X=\mathcal{R}^{-1}$,
				$X_{0}=X\cap(X_{0,1}\times X_{0,2})$,
				$U=U_{1}$, 
				$(x_{1},x_{2})\rTo^{u}(x_{1}^{\prime},x_{2}^{\prime})$ if $x_{1}\rTo_{1}^{u}x_{1}^{\prime}$ and $x_{2}\rTo_{2}^{u}x_{2}^{\prime}$,
				$Y=Y_{1}$ and 
				$H(x_{1},x_{2})=H_{1}(x_{1})$ for any $(x_{1},x_{2})\in X$.
			\end{definition}
			We conclude with a useful technical lemma.
			
			\begin{lemma}
				\cite{paulo}
				Let \mbox{$S_{i}=(X_{i},X_{0,i},U_{i},\rTo_{i},Y_{i},H_{i})$} ($i=1$, $2$, $3$) be (pseudo)metric systems with the same output sets $Y_{1}=Y_{2}=Y_{3}$ and (pseudo)metric $d$. Then, the following statements hold: 
				\begin{itemize}
					\item
					[(i)] for any $\varepsilon_{1}\!\!\leq \!\varepsilon_{2}$, $S_{1}\preceq_{\varepsilon_{1}}^{(s,\alt)}S_{2}$ implies \mbox{$S_{1}\preceq_{\varepsilon_{2}}^{(s,\alt)}S_{2}$}; 
					\item [(ii)] if $S_{1}\preceq^{(s,\alt)}_{\varepsilon_{12}}S_{2}$ with relation $\mathcal{R}_{12}$ and $S_{2}\preceq^{(s,\alt)}S_{3}$ with relation $\mathcal{R}_{23}$ then $S_{1}\preceq^{(s,\alt)}_{\varepsilon_{12}+\varepsilon_{23}}S_{3}$ with relation $\mathcal{R}_{12} \circ \mathcal{R}_{23}$; 
					\item [(iii)] for any $\varepsilon\in\mathbb{R}^{+}_{0}$ and any strong $A\varepsilon A$ simulation relation $\mathcal{R}$ from $S_2$ to $S_1$, $S_{1} \times^{\mathcal{R}}_{\varepsilon} S_{2} \preceq_{\varepsilon} S_{2}$.
				\end{itemize}
				\label{lemma1}
			\end{lemma}

\end{document}